\documentclass[11pt]{article}

\usepackage[top=2cm, bottom=2cm, left=2cm, right=2cm, includefoot]{geometry}
\usepackage{authblk}
\usepackage{hyperref}
\usepackage{cleveref}
\usepackage{graphicx}
\usepackage{amsthm}
\usepackage{amsmath}
\usepackage{url}
\usepackage[algoruled,linesnumbered,algo2e,lined,noend]{algorithm2e}
\usepackage{comment}

\usepackage{thmtools}
\usepackage{thm-restate}
\usepackage{color}
\newtheorem{observation}{Observation}
\newtheorem{definition}{Definition}
\newtheorem{proposition}{Proposition}

\newcommand{\mi}[1]{\mathit{#1}}

\newcommand{\para}[1]{{#1}}

\newcommand{\str}[2]{{\mathsf{str}({#1},{#2}})}

\newcommand{\pre}[1]{\mathsf{prev}({#1})}

\newcommand{\calN}{{\cal N}}

\newcommand{\STree}{\mathsf{STree}}
\newcommand{\STrie}{\mathsf{STrie}}
\newcommand{\LST}{\mathsf{LST}}
\newcommand{\PSTree}{\mathsf{PSTree}}
\newcommand{\PSTrie}{\mathsf{PSTrie}}
\newcommand{\Recode}{\mathsf{Re}}

\newcommand{\Shrink}{\mathsf{sl}}
\newcommand{\Slink}{\mathsf{SL}}
\newcommand{\Flink}{\mathsf{FL}}

\newcommand{\PLST}{\mathsf{PLST}}

\newcommand{\Child}{\mathsf{child}}

\SetKwProg{Fn}{Function}{}{end}

\newcommand{\reencode}[1]{\langle #1 \rangle}
\newcommand{\PSub}{\mathsf{PrevSub}}

\newcommand{\vtx}[1]{#1}

\newtheorem{theorem}{Theorem}
\newtheorem{lemma}{Lemma}

\begin{document}

\title{An Extension of Linear-size Suffix Tries for Parameterized Strings}
\author[1]{Katsuhito Nakashima}
\author[1]{Diptarama Hendrian}
\author[1]{Ryo Yoshinaka}
\author[1]{Ayumi Shinohara}

\affil[1]{Graduate School of Information Sciences, Tohoku University, Japan}

\date{}

\maketitle            

\begin{abstract}
%#! platex main
%In this paper, we propose a new indexing structure for parameterized strings which we call \emph{PLSTs}, by combining linear-size suffix tries and parameterized suffix trees.
In this paper, we propose a new indexing structure for parameterized strings which we call \emph{PLSTs}, by generalizing linear-size suffix tries for ordinary strings.
Two parameterized strings are said to match if there is a bijection on the symbol set that makes the two coincide. 
PLSTs are applicable to the parameterized pattern matching problem,
%\footnote{In the previous version (https://arxiv.org/abs/1902.00216v1), we claimed that our matching algorithm could run in linear time in the pattern size but the sketched proof idea was in error.},
 which is to decide whether the input parameterized text has a substring that matches the input parameterized pattern.
The size of PLSTs is linear in the text size, 
with which our algorithm solves the parameterized pattern matching problem in linear time in the pattern size.
PLSTs can be seen as a compacted version of parameterized suffix tries and a combination of linear-size suffix tries and parameterized suffix trees.
We experimentally show that PLSTs are more space efficient than parameterized suffix trees for highly repetitive strings. 
%However, it is open whether one can solve the parameterized pattern matching problem in linear time in the pattern size using our proposed data structure.

%Parameterized linear-size suffix tries suppress non-branching nodes of the original trie and use suffix links to recover edge labels.
%Unlike parameterized suffix trees, parameterized linear-size suffix tries do not have references to the input text. 

\end{abstract}

%#! platex main
\section{Introduction}
The pattern matching problem is to check whether a pattern string occurs in a text string or not.
%Formally, given a text string $T$ and a pattern string $P$ over an alphabet $\Sigma$, output all positions at which $P$ occurs in $T$.
To efficiently solve the pattern matching problem, a numerous number of text indexing structures have been proposed.
%A suffix trie is a basic tree indexing structure but it requires $O(n^2)$ space where $n$ is the size of an input text.
Suffix trees %and suffix arrays
 are most widely used data structures
and provide many applications including several variants of pattern matching problems~\cite{JEWELS,GUS}.
They can be seen as a compacted type of suffix tries, where two branching nodes that have no other branching nodes between them in a suffix trie are directly connected in the suffix tree.
The new edges have a reference to an interval of the text so that the original path label of the suffix trie can be recovered.
Recently, Crochemore et al.~\cite{crochemore2016linear} proposed a new indexing structure, called a \emph{linear-size suffix trie (LST)}, which is another compacted variant of a suffix trie.
An LST replaces paths consisting only of non-branching nodes by edges like a suffix tree, but the original path labels are recovered by referring to other edge labels in the LST itself unlike suffix trees.
LSTs use less memory space than suffix trees for indexing the same highly repetitive strings~\cite{takagi2017linear}.
%like fibonacci string...
LSTs may be used as an alternative of suffix trees for various applications, like computing the longest common substrings, not limited to the pattern matching problem.
% They expect that LSTs may have potential to be generalized to variants of string matching problems just like suffix trees do.
% LSTs need not to keep the original text $T$ by additional nodes and suffix links. LSTs use experimentally less memory than suffix trees do.
% Moreover, LSTs provide efficient search of patterns by preprocessing the text string, similarly to suffix tree and suffix tries.

%A most representative application of indexing structures is pattern matching.
On the other hand, different types of pattern matching have been proposed and intensively studied.
The variant this paper is concerned with is the \emph{parameterized pattern matching problem}, introduced by Baker~\cite{PMA}.
 %that focuses on a structure of strings
Considering two disjoint sets of symbols $\Sigma$ and $\Pi$,
 we call a string over $\Sigma \cup \Pi$ a {\em parameterized string (p-string)}.
In the parameterized pattern matching problem,
given p-strings $\para{T}$ and $\para{P}$,
we must check whether substrings of $\para{T}$
that can be transformed into $\para{P}$ by applying a one-to-one function
that renames symbols in $\Pi$.
The parameterized pattern matching is motivated
by applying to the software maintenance~\cite{PM,PMA},
the plagiarism detection~\cite{EPM},
the analysis of gene structure~\cite{SST},
and so on.
Similarly to the basic string matching problem,
several indexing structures that support the parameterized pattern matching
have been proposed,
such as parameterized suffix trees~\cite{PMA}, structural suffix trees~\cite{SST}, parameterized suffix arrays~\cite{PSC2008-8,I2009}, and parameterized position heaps~\cite{diptarama2017position,fujisato2018right}.

In this paper, we propose a new indexing structure for p-strings, which we call \emph{PLST}.
A PLST is a tree structure that combines a linear-size suffix trie and a parameterized suffix tree for prev-encoded~\cite{PMA} suffixes of a p-string.
We show that the size of a PLST is $O(n)$ and %present our $O(n)$ size data structure and
 give an algorithm for the parameterized pattern matching problem for given a pattern and a PLST,
to find the occurrences of a pattern in the text,
that runs in $O(m)$ time, 
where $n$ is the length of the text and $m$ is the length of the pattern.
Furthermore, we experimentally show that PLSTs are more space efficient than parameterized suffix trees for highly repetitive strings such as Fibonacci strings.
%It is open whether our matching algorithm runs in linear time in the pattern size.
% Thue-Morse strings and Period-doubling strings.

%We assume that the alphabet of the symbols in  text is constant-size, although the result can be extended to any alphabet by multiplying the time complexity by a factor of $\log (\Sigma \cup \calN)$, as it is customary in the literature.

%The rest of this paper is organized as follows. In Section~\ref{sec:Prelim}, we give preliminaries
%on the problem. In addition, we recall string data structures such as linear-size suffix tries and parameterized suffix trees. 
%In Section~\ref{sec:PLST}, we describe our indexing structure and algorithm for the parameterized pattern matching.
%Finally, we conclude our work and discuss future work in Section~\ref{sec:conclusion}.

%Throughout this paper,
%we assume that the size of input alphabet is constant. 

\section{Preliminaries}\label{sec:Prelim}

\subsection{Basic definitions and notation}
We denote the set of all non-negative integers by $\calN$.
% What does the space here mean?
Let $\Delta$ be an alphabet.
For a string $\para{w}=\para{xyz} \in \Delta^*$, $\para{x}$, $\para{y}$, and $\para{z}$ are called 
\emph{prefix}, \emph{substring}, and \emph{suffix} of $\para{w}$, respectively.
The length of $\para{w}$ is denoted by $|\para{w}|$ and
 the $i$-th symbol of $\para{w}$ is denoted by $\para{w}[i]$ for $1 \leq i \leq |\para{w}|$.
The substring of $\para{w}$ that begins at position $i$ and ends at position $j$ 
is denoted by $\para{w}[i:j]$ for $1 \leq i \leq j \leq |\para{w}|$.
For convenience,
we abbreviate $\para{w}[1:i]$ to $\para{w}[:i]$ and $\para{w}[i:|w|]$ to $\para{w}[i:]$ for $1 \leq i \leq |\para{w}|$.
The empty string is denoted by $\varepsilon$, that is $|\varepsilon|=0$.
Moreover, let $\para{w}[i:j] = \varepsilon$ if $i > j$.
For a string $u$ and an extension $uv$, we write $\str{u}{uv}=v$.
%For a nonempty string $au$ with $a \in \Delta$ and $u \in \Delta^*$, the string $u$ obtained by removing the first symbol is denoted by $\Shrink(au)$.

Throughout this paper, we fix two alphabets $\Sigma$ and $\Pi$.
We call elements of $\Sigma$ \emph{constant} symbols
and those of $\Pi$ \emph{parameter} symbols.
An element of $\Sigma^{*}$ is called a \emph{constant string}
and that of $(\Sigma \cup \Pi)^{*}$ is called a \emph{parameterized string},
or \emph{p-string} for short.
We assume that the size of $\Sigma$ and $\Pi$ are constant.

Given two p-strings $w_1$ and $w_2$ of length $n$,  
$w_1$ and $w_2$ are a \emph{parameterized match} \emph{(p-match)}, denoted by $w_1 \approx w_2$,
if there is a bijection $f$ on $\Sigma \cup \Pi$ such that $f(a) = a$ for any $a \in \Sigma$ and $f(w_1[i]) = w_2[i]$ for all $1 \leq i \leq n$~\cite{PMA}.
We can determine whether $w_1 \approx w_2$ or not by using an encoding called \emph{prev-encoding} defined as follows.

\begin{definition}[Prev-encoding~\cite{PMA}]
For a p-string $w$ of length $n$ over $\Sigma \cup \Pi$,
the \emph{prev-encoding} for $w$, denoted by $\pre{w}$,
is defined to be a string over $\Sigma \cup \calN$  of length $n$ such that for each $1 \leq i \leq n$,

\begin{align*}
	\pre{w}[i] = \begin{cases}
    \para{w}[i] & \text{if }\para{w}[i] \in \Sigma , \\
    0           & \text{if }\para{w}[i] \in \Pi \text{ and } \para{w}[i] \neq \para{w}[j] \text{ for } 1 \le j < i,\\
    i-k  		& \text{if }\para{w}[i] \in \Pi \text{ and } k = \max\{j \mid \para{w}[j]=\para{w}[i] \text{ and } 1 \le j < i\}.
  \end{cases}
\end{align*}

We call strings over $\Sigma \cup \calN$ \emph{pv-strings}.
\end{definition}

For any p-strings $w_1$ and $w_2$,  $w_1 \approx w_2$ if and only if $\pre{w_1}=\pre{w_2}$.
For example, given $\Sigma = \{{\tt a, b}\}$ and $\Pi = \{{\tt u, v, x, y}\}$,
$s_1 = {\tt uvvvauuvb}$ and $s_2 = {\tt xyyyaxxyb}$ are p-matches 
by $f$ such that $f(\mathtt{u})=\mathtt{x}$ and $f(\mathtt{v})=\mathtt{y}$,
where $\pre{s_1} = \pre{s_2} = 0011{\tt a}514{\tt b}$.

We define \emph{parameterized pattern matching} as follows.
%\begin{definition}[Parameterized pattern matching~\cite{PST}]
%Given two p-strings, text $T$ and pattern $P$,
%find all positions $i$ in $T$ such that $T[i:i+|P|-1]\approx P$. 
%\end{definition}
\begin{definition}[Parameterized pattern matching~\cite{PMA}]
Given two p-strings, text $T$ and pattern $P$, decide whether $T$ has a substring that p-matches $P$.
\end{definition}
For example, considering a text $T={\tt auvaubuavbv}$ and a pattern $P = {\tt xayby}$ over $\Sigma = \{{\tt a, b}\}$ and $\Pi = \{{\tt u, v, x, y}\}$,
$T$ has two substrings $T[3:7]={\tt vaubu}$ and $T[7:11] = {\tt uavbv}$ that p-match $P$.

Throughout this paper, we assume that a text $T$ ends with a sentinel symbol $\texttt{\$} \in \Sigma$, which occurs nowhere else in $T$.

\subsection{Suffix tries, suffix trees, and linear-size suffix tries} \label{ST}
This subsection briefly reviews tree structures for indexing all the substrings of a constant string $T \in \Sigma^*$.

The suffix trie $\STrie(T)$ is a tree with nodes corresponding to all the substrings of $T$.
Figure~\ref{fig:Strie_LST}~(a) shows an example of a suffix trie. 
Throughout this paper, we identify a node with its corresponding string for explanatory convenience.
Note that each node does not explicitly remember its corresponding string.
For each nonempty substring $ua$ of $T$ where $a \in \Sigma$, we have an edge from $\vtx{u}$ to $\vtx{ua}$ labeled with $a$.
Then by reading the labels on the path from the root to a node $u$, one can obtain the string $u$ the node corresponds.
Then the path label from the node $\vtx{u}$ to a descendant $\vtx{uv}$ is $\str{u}{uv}=v$ for $u,v \in \Sigma^*$.
Since there are $\Theta(|T|^2)$ substrings of $T$, the size of $\STrie(T)$ is $\Theta(|T|^2)$. 

The suffix tree $\STree(T)$ is a tree obtained from $\STrie(T)$ by removing all non-branching internal nodes and replacing each path with no branching nodes by a single edge whose label refers to a corresponding interval of the text $T$.
That is, the label on the edge $(\vtx{u},\vtx{v})$ is a pair $(i,j)$ such that $T[i:j] = \str{u}{v}$.
Since there are at most $O(|T|)$ branching nodes,  the size of $\STree(T)$ is $\Theta(|T|)$. 

An important auxiliary map on nodes is called \emph{suffix links}, denoted by $\Slink$, which is defined by $\Slink(\vtx{aw}) = \vtx{w}$ for each node $aw$ with $a \in \Sigma$ and $w \in \Sigma^*$.

\begin{figure}[t]
	\centering
	\begin{minipage}[t]{0.48\hsize}
		\centering
		\includegraphics[scale=0.115]{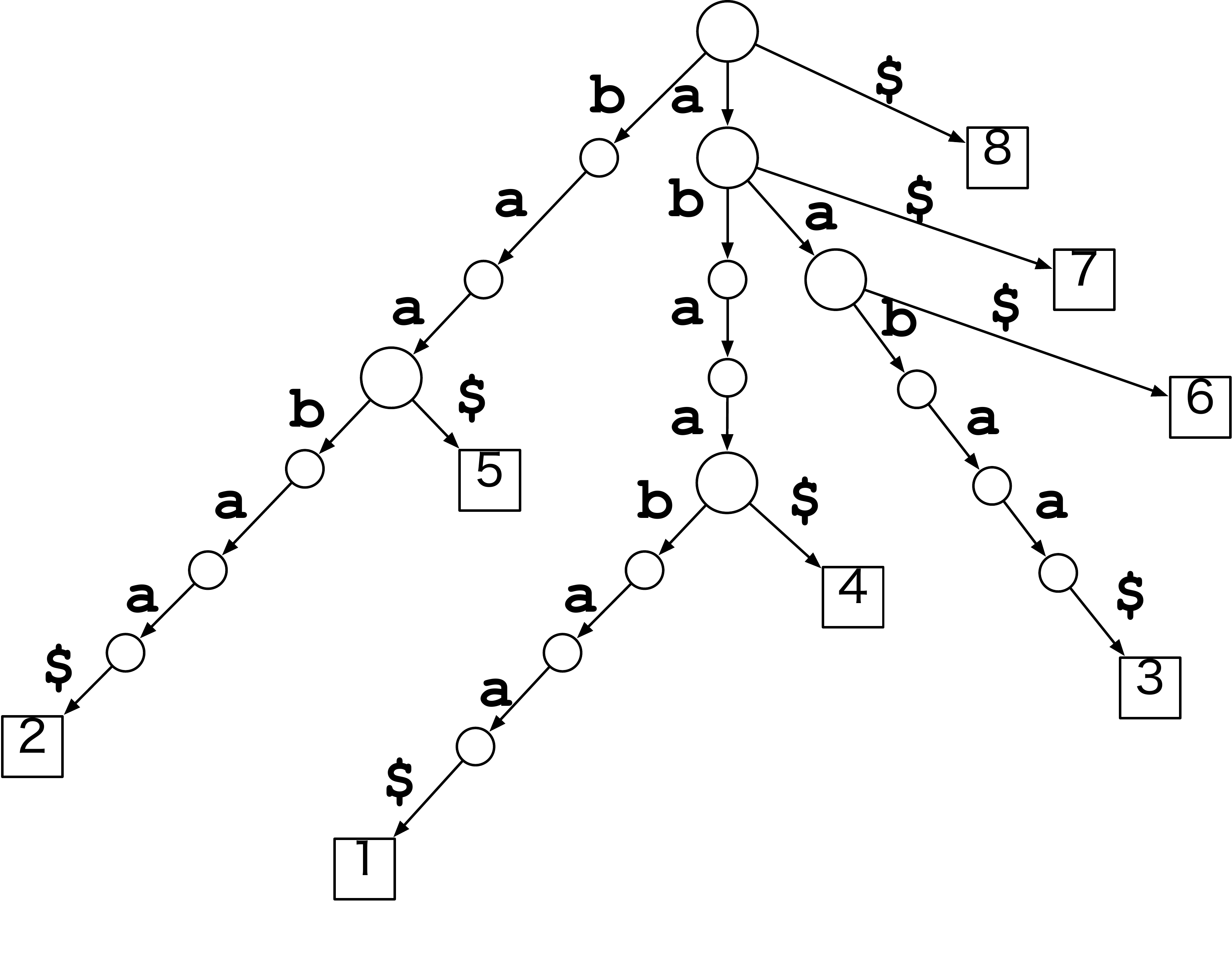}\\
		\ \ \ \scriptsize{(a)}
	\end{minipage}
	\begin{minipage}[t]{0.48\hsize}
		\centering
		\includegraphics[scale=0.115]{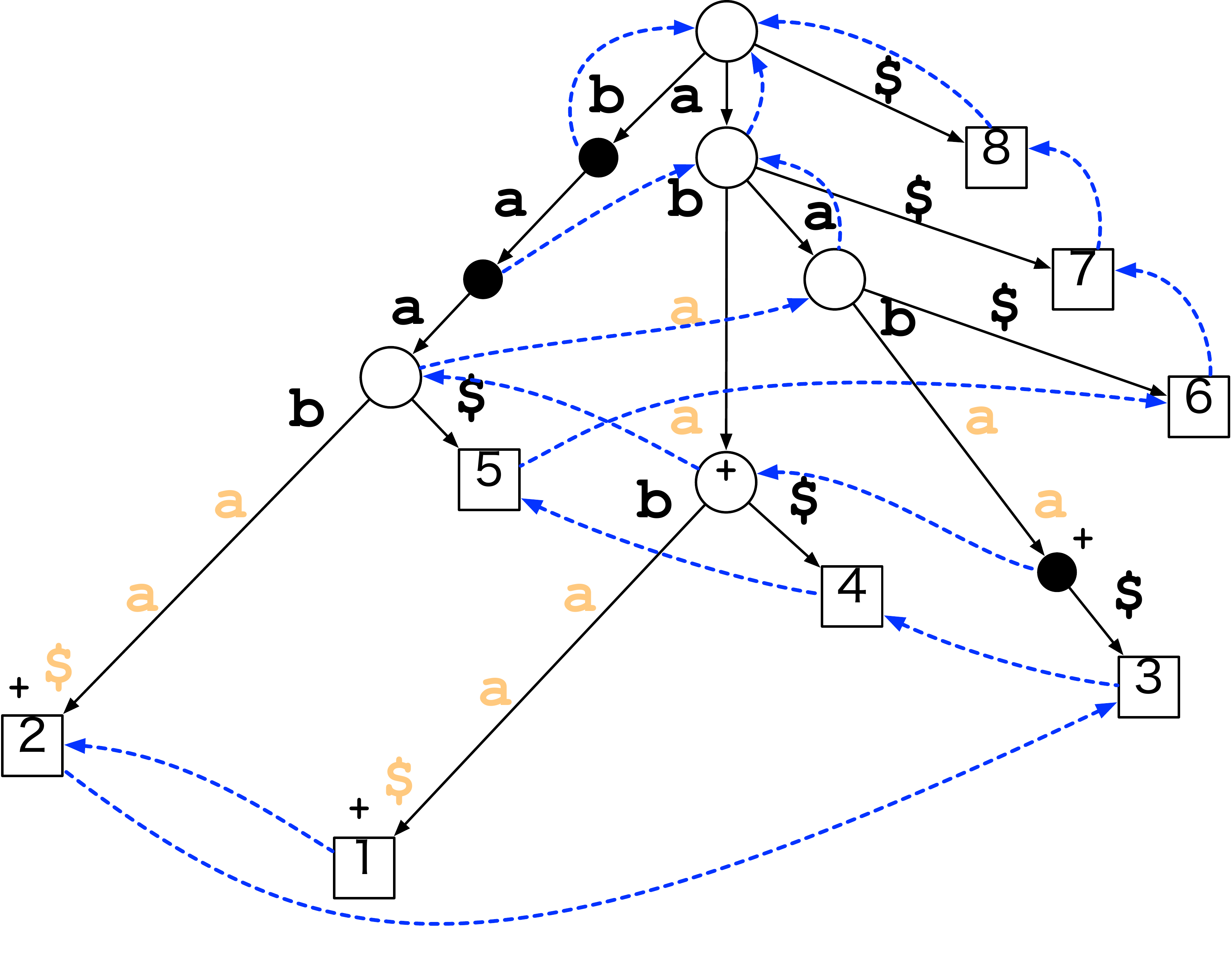}\\
		\ \ \ \scriptsize{(b)}
	\end{minipage}
	\caption{
		(a) The suffix trie for $T={\tt abaabaa}{\texttt \$}$.
		(b) The LST for $T$.
		Solid and broken arrows represent the edges and suffix links, respectively. 
		%Black and orange symbols attached to each edge are the first (real) and the following (implicit) symbols of the original edge label. 
		The LST keeps the first symbol (black) on each edge, while the succeeding symbols (orange) are discarded. 
		Big white and small black circles represent nodes of Type~1 and Type~2, respectively. 
		The `$+$' signs represent the 1-bit flag. 
		If a node $v$ has `$+$' sign, the edge $(u, v)$ has a path label of length greater than 1 in $\STrie(T)$ where $u$ is the parent node of $v$ in $\LST(T)$.
	}
	\label{fig:Strie_LST}
\end{figure}

The \emph{linear-size suffix trie (LST)}~\cite{crochemore2016linear} $\LST(T)$ of a string $T$ is another compact variant of a suffix trie (see Figure~\ref{fig:Strie_LST}~(b)). % like a suffix tree, but it dose not have references to a text $T$.
An LST suppresses (most) non-branching nodes and replaces paths with edges like a suffix tree, but the labels of those new edges do not refer to intervals of the input text.
Each edge $(\vtx{u},\vtx{v})$ retains only the first symbol $\str{u}{v}[1]$ of the original path label $\str{u}{v}$.
To recover the original label $\str{u}{v}$, we refer to another edge or a path in the LST itself following a suffix link,
using the fact that $\str{u}{v} = \str{\Slink(u)}{\Slink(v)}$. % where $\vtx{u'}=\Slink(\vtx{u}) $ and $\vtx{v'}=\Slink(\vtx{v})$.  
The reference will be recursive, but eventually one can regain the original path label by collecting those retained symbols.
For this sake, $\LST(T)$ keeps some non-branching internal nodes from $\STrie(T)$ and thus it may have more nodes than $\STree(T)$, but still the size is linear in $|T|$.
The nodes of $\LST(T)$ consist of those of $\STree(T)$ and non-branching node whose suffix links point at a branching node.
We call the former \emph{Type~1} and the latter \emph{Type~2}.
%We classify the nodes of $\STrie(T)$ into Type~1, Type~2 and others as follows, among which Type~1 nodes are exactly those of $\STree(T)$ and in addition Type~2 nodes constitute $\LST(T)$.
%\begin{enumerate}
%	\item Type~1 nodes are either a leaf or a branching node. %That is, those are exactly the nodes in the suffix tree of $T$.
%	\item Type~2 nodes are non-branching internal nodes whose suffix link points at a Type~1 node. %trie nodes whose suffix link point a node in the suffix tree of $T$.
%\end{enumerate}
Each edge $(\vtx{u},\vtx{v})$ has a 1-bit flag that tells whether $|v|-|u|=1$.
If it is the case, one knows the complete label $\str{u}{v}=\str{u}{v}[1]$.
Otherwise, one needs to follow the suffix link to regain the other symbols.
An LST uses suffix links to regain the original path label in the suffix trie.
If we had only Type~1 nodes, for some edge $(\vtx{u},\vtx{v})$, there may be a branching node between $\Slink(\vtx{u})$ and $\Slink(\vtx{v})$, which makes it difficult to regain the original path label.
Having Type~2 nodes, there is no branching node between $\Slink(\vtx{u})$ and $\Slink(\vtx{v})$ for every edge $(\vtx{u}, \vtx{v})$.  Then it is enough to go straight down from $\Slink(\vtx{u})$ to regain the original path label.

\subsection{Parameterized suffix tries and parameterized suffix trees}\label{PST}

\begin{figure}[t]
	\centering
	\begin{minipage}[t]{0.45\hsize}
		\centering
		\includegraphics[scale=0.10]{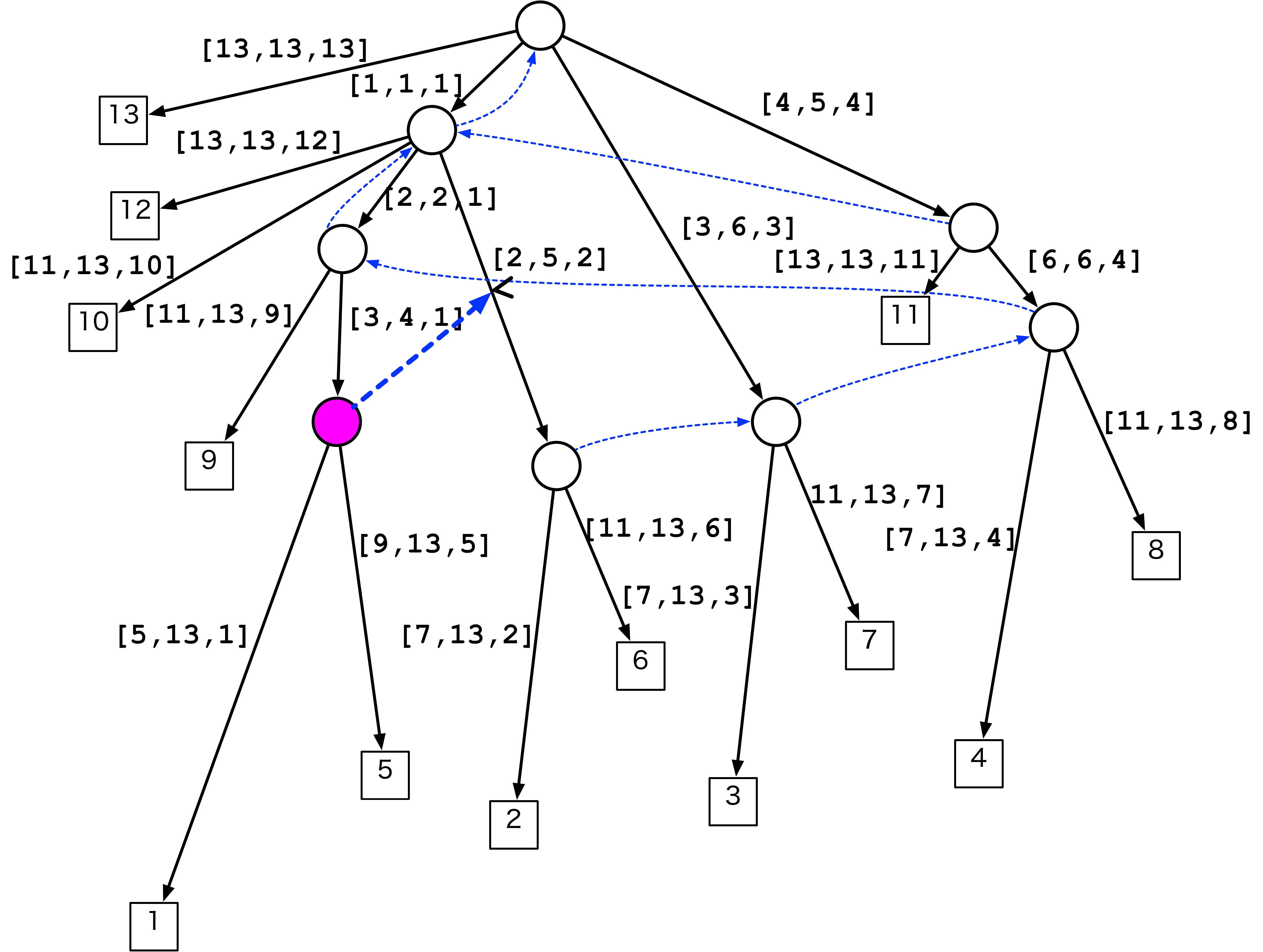}\\
		\includegraphics[scale=0.25]{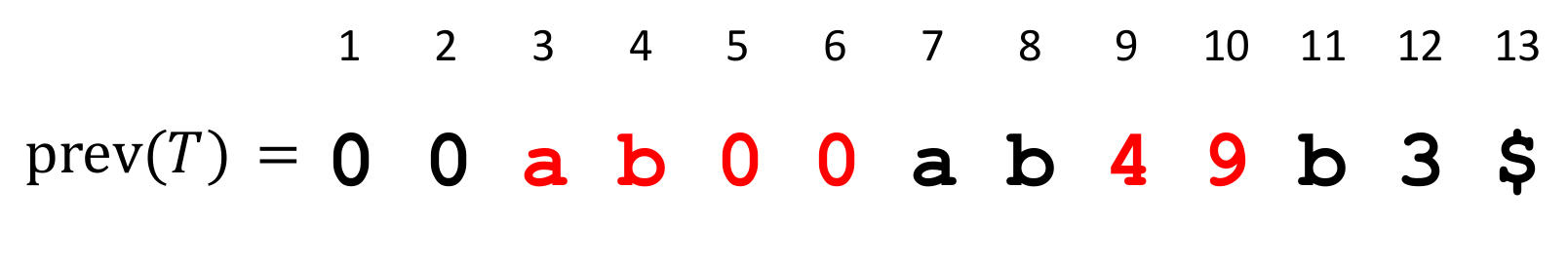}\\
		\ \ \ \scriptsize{(a)}
	\end{minipage}
	\begin{minipage}[t]{0.45\hsize}
		\centering
		\includegraphics[scale=0.10]{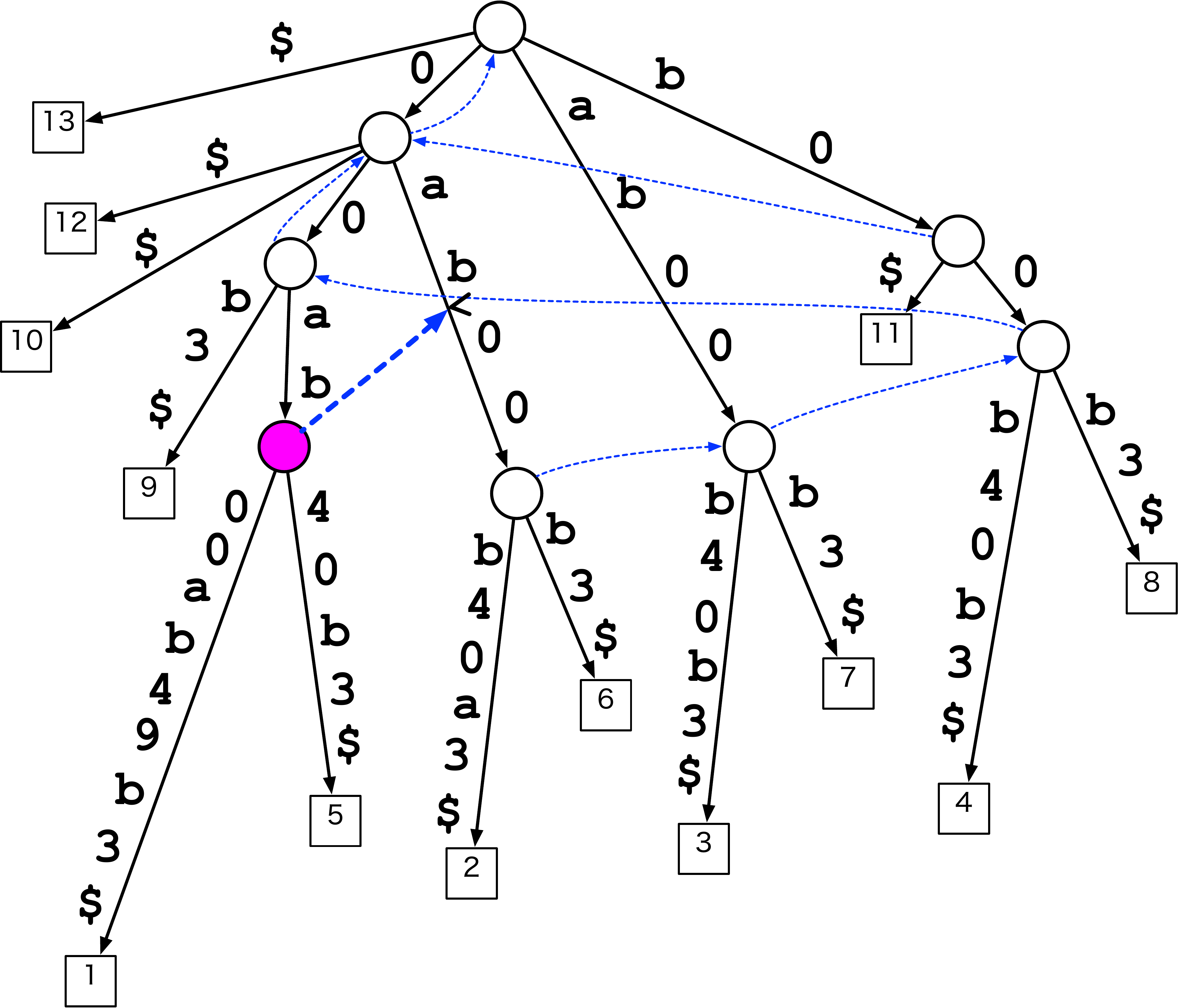}\\		
 		\ \ \ \scriptsize{(b)}
	\end{minipage}
	\caption{
		(a) The parameterized suffix tree $\PSTree(T)$ for $T={\tt xyabzwabzxbz}{\texttt \$}$ where $\Sigma=\{ {\tt a, b, }{\texttt   \$}\}$ and $\Pi=\{{\tt x,y,z,w}\}$.
		(b) The referenced substrings are shown on edges.
		Broken blue arrows denote suffix links. Some suffix links do not point to a branching node shown with a bold broken arrow. % and circle are bad suffix link and bad node	
	}
	\label{fig:p-string_PST}
\end{figure}

For a p-string $T \in (\Sigma \cup \Pi)^*$, a \emph{prev-encoded substring (pv-substring)} of $T$ is the {prev-encoding} $\pre{w}$ of a substring $w$ of $T$. 
The set of pv-substrings of $T$ is denoted by $\PSub(T)$.

A parameterized suffix trie of $T$, denoted by $\PSTrie(T)$, is the trie that represents all the pv-substrings of $T$. The size of $\PSTrie(T)$ is $\Theta(|T|^2)$. 

For a pv-string $u \in (\Sigma \cup \calN)^*$,
the $k$-\emph{re-encoding} for $u$, denoted by $\reencode{u}_k$,
is defined to be the pv-string of length $|u|$ such that for each $1 \leq i \leq |u|$,
\[
  \reencode{u}_k[i] = \begin{cases}
    0           & \text{if }\para{u}[i] \in \calN \text{ and } \para{u}[i] \geq i-k+1 , \\
    \para{u}[i] & \text{otherwise.}
  \end{cases}
\]
When $k = 1$, we omit $k$.
We then have $\reencode{\pre{w}[i:j]}=\pre{w[i:j]}$ for any p-string $w \in (\Sigma \cup \Pi)^*$ and $i,j \le |w|$.

Usually suffix links are defined on nodes of suffix trees, but it is convenient to have ``implicit suffix links'' on all nodes except the root of $\STrie(T)$, i.e., all the nonempty substrings of $T$, as well.
For a nonempty pv-string $u \in (\Sigma \cup \calN)^+$, let $\Shrink(u)$ denote the re-encoding $\reencode{u[2:]}$ of the string obtained by deleting the first symbol.
This operation on strings will define real suffix links in indexing structures for parameterized strings based on parameterized suffix tries.
Differently from constant strings, $u \in \PSub(T)$ does not necessarily imply $u[2:] \in \PSub(T)$. 
What we actually have is $\Shrink(u)=\reencode{u[2:]} \in \PSub(T)$.
%For each pv-substring $u$ of $T$, we define the \emph{suffix link} of $\vtx{u}$ by $\Slink(\vtx{u}) = \vtx{\Shrink(u)}$.

A \emph{parameterized suffix tree (p-suffix tree)}~\cite{PMA} of $T$, denoted by $\PSTree(T)$, is a compacted variant of the parameterized suffix trie.
Figure~\ref{fig:p-string_PST} shows an example of a p-suffix tree. 
Like the suffix tree for a constant string over $\Sigma$, $\PSTree(T)$ is obtained from $\PSTrie(T)$ by removing non-branching internal nodes and giving each edge as a label a reference to some interval of the prev-encoded text $\pre{T}$.
%In many cases, 
The reference is represented by a triple $(i, j, k)$ of a text start position, end position, and suffix number, which refers to the pv-string $\reencode{\pre{T}[k:]}[i:j]$. 

%Lee et al.~\cite{lee2011line} showed that $\PSTree(T)$ can be built online in randomized $O(|T|)$ time by using suffix links, which connect nodes $u$ and $\Shrink(u)$.
%In a (non-parameterized) suffix tree, the suffix link of a branching node necessarily points to a branching node.
%That is, if $u$ is a node, then so is $\Shrink(u)$.
%However, in a parameterized suffix tree, there are branching nodes whose ``suffix links'' do not point to a branching node. 
%Figure~\ref{fig:p-string_PST} shows an example, where the node $\vtx{00\mathtt{ab}}$ (red circle) is a branching node in $\PSTree(T)$ but $\Shrink(\vtx{00\mathtt{ab}})=\vtx{0\mathtt{ab}}$ is not.
%This does not matter for pattern matching using p-suffix trees. 
%%although suffix links are important for construction.
%However, it is critical to designing a parameterized counterpart of LST,
%since we need to recursively follow suffix links to recover the original label.
% %to regain the original path label of a edge $(u, v)$ from $(\Slink(u), \Slink(v))$.
%%We will discuss this point in more detail in the next section.
%Therefore, We propose a compromise between a p-suffix tree and an LST in the next section.
\section{PLSTs}\label{sec:PLST}

\newcommand{\ta}{\mathtt{a}}
	\newcommand{\tx}{\mathtt{x}}
	\newcommand{\ty}{\mathtt{y}}
	\newcommand{\tz}{\mathtt{z}}
	\newcommand{\z}{\mathrm{0}}
	\newcommand{\s}{\mathrm{6}}

We now introduce our indexing tree structures for p-strings, which we call \emph{PLSTs}, based on LSTs and p-suffix trees reviewed in Sections~\ref{ST} and \ref{PST}.
%An example of a PLST compared with a p-suffix tree is shown in Figure~\ref{fig:PLST}.
There are two difficulties in extending LSTs to deal with p-strings.
Figure~\ref{fig:PLST}(a) shows the LST-like structure obtained from $\PSTrie(T)$ in the same way as $\LST(T)$ is obtained from $\STrie(T)$.
We want to know $\str{u}{v}$ for an edge $(u, v)$ by ``reduction by suffix links'', but
\begin{enumerate}
	\item it is \emph{not} necessarily that $\str{u}{v} = \str{\Shrink(u)}{\Shrink(v)}$,
	\item there can be a branching node $u$ of $\PSTrie(T)$ such that $\Shrink(u)$ is not branching.
\end{enumerate}
An example edge $(u,v)$ exhibiting the first difficulty consists of $u=00$ and $v=00\mathtt{b}3\texttt{\$}$, where $\str{u}{v} = \texttt{b}3\texttt{\$}$ but $\str{\Shrink(u)}{\Shrink(v)} = \texttt{b}0\texttt{\$}$.
This is caused by the fact that $\Shrink(u) = \reencode{u[2:]}$ rather than $\Shrink(u) = u[2:]$.
Then, the path label $\str{\Shrink(u)}{\Shrink(v)}$ referenced by the suffix link may not give exactly what we want.
We solve this problem by giving the node $v$ a ``re-encoding sign'' with which one can recover $\str{u}{v}$ from $\str{\Shrink(u)}{\Shrink(v)}$.
An example node for the second case is $\vtx{00\mathtt{ab}}$.
This is a branching node but $\reencode{\Shrink(\vtx{00\mathtt{ab}})}=\vtx{0\mathtt{ab}}$ does not appear as a node.
To handle this case, we simply refer to the corresponding interval of the original text $T$ by keeping the necessary subsequence,
where, as we will observe in experiments, the necessary subsequence is tend to be rather small.
Our proposed structure PLST is shown in Figure~\ref{fig:PLST}(b). 
In what follows we explain PLSTs.
%In a (non-parameterized) suffix tree, the suffix link of a branching node necessarily points to a branching node.
%That is, if $u$ is a node, then so is $\Shrink(u)$.
%An example of the second case is observed in Figure~\ref{fig:p-string_PST}, where the node $\vtx{00\mathtt{ab}}$ (red circle) is a branching node in $\PSTree(T)$ but $\Shrink(\vtx{00\mathtt{ab}})=\vtx{0\mathtt{ab}}$ is not.
%The second one is critical to regain the original path label in the suffix trie.
%If we do not have $\Shrink(u)$ in our indexing structure, we cannot use the technique of ``reduction by suffix links'' to regain the original path label.

\subsection{Definition and properties of PLSTs}

\begin{figure}[t]
	\centering
	\begin{minipage}[t]{0.426\hsize}
		\centering
		\includegraphics[scale=0.12]{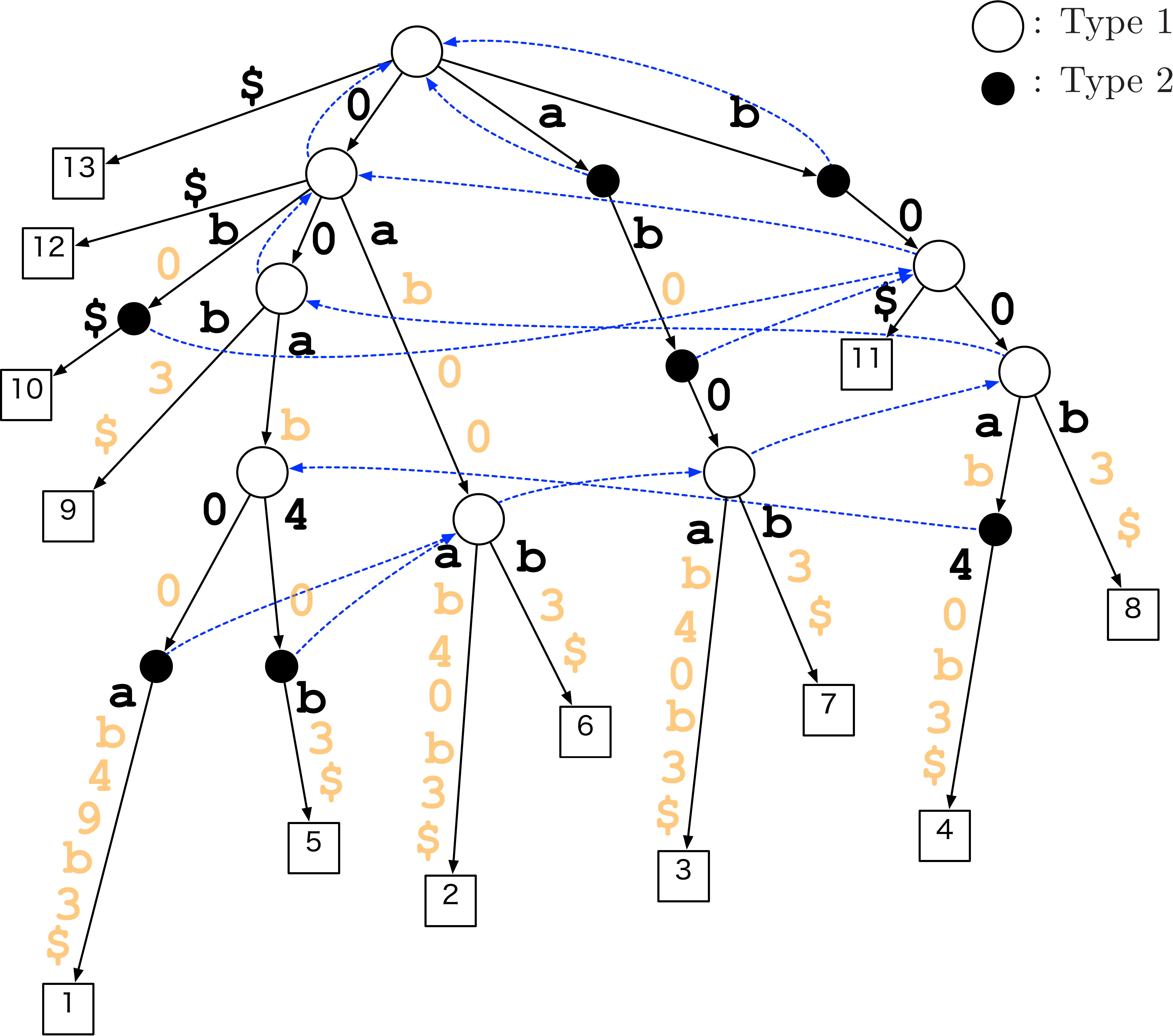}\\
		\ \ \ \scriptsize{(a)}
	\end{minipage}
	\begin{minipage}[t]{0.566\hsize}
		\centering
		\includegraphics[scale=0.12]{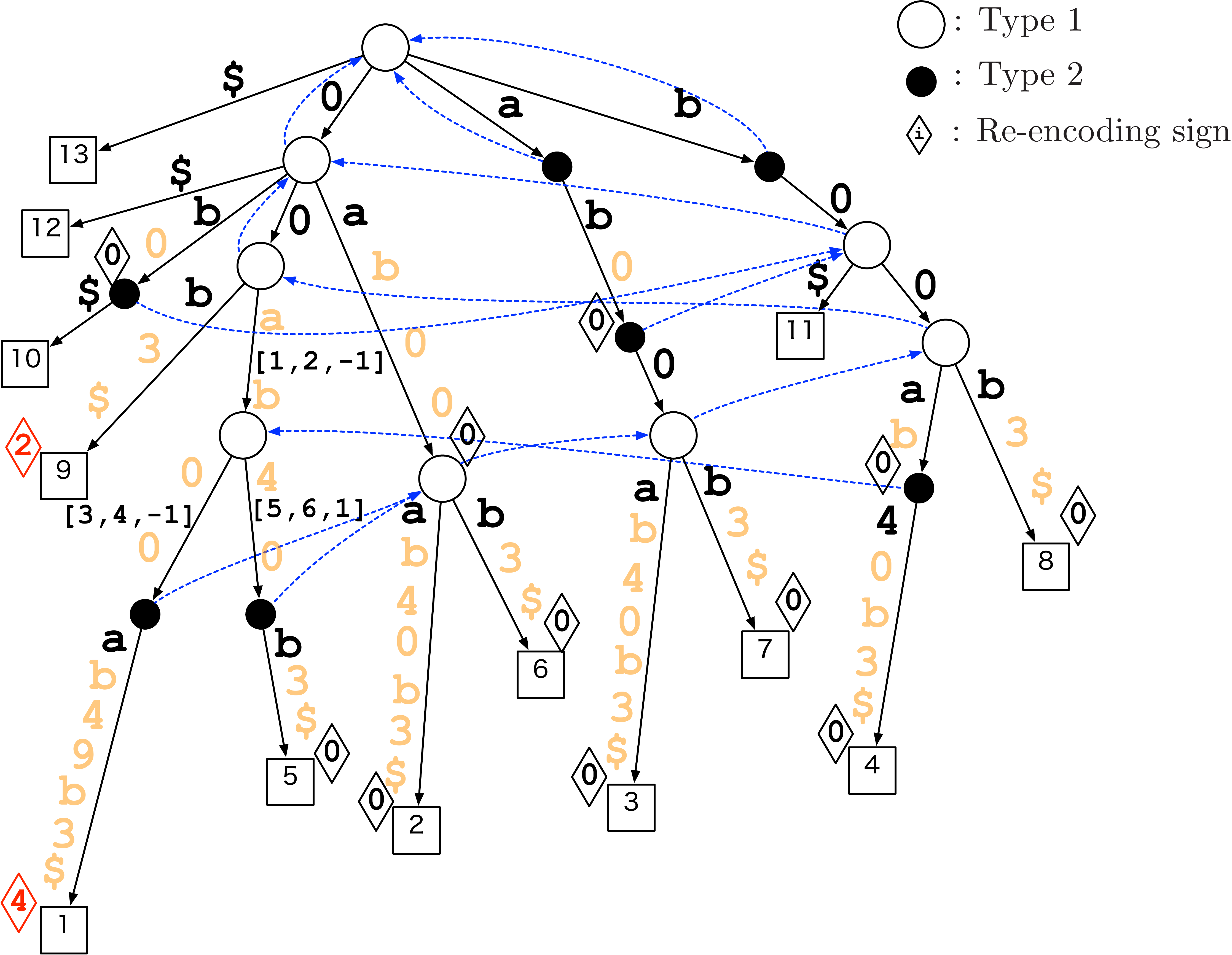}\\
		\includegraphics[scale=0.3]{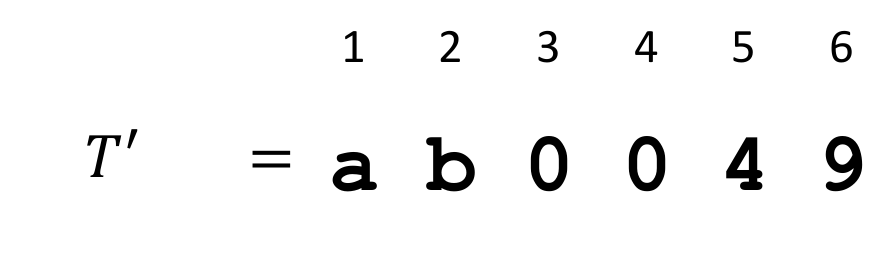}\\
		\ \ \ \scriptsize{(b)}
	\end{minipage}
	\caption{
		(a) LST-like structure for $\PSub(T)$ and (b) the PLST $\PLST(T)$ for $T={\tt xyabzwabzxbz}{\texttt \$}$ 
%		($\pre{T} = {\tt 00ab00ab49a3\$}$) 
		where $\Sigma=\{ {\tt a, b,}{\texttt \$}\}$ and $\Pi=\{{\tt x, y, z, w}\}$. 
		The triples of integers are reference to the text.
%		 and are edge labels. 
 White and black circles represent nodes of Type~1 and 2, respectively. 
 The numbers in rhombus represent re-encoding signs.
 The PLST keeps only the first symbol (black) on each edge, while the succeeding symbols (orange) are discarded. 
 The reference text $T' = {\tt ab0049}$ is shorter than the original $T$.
	}
	\label{fig:PLST}
\end{figure}

Let $U = \PSub(T)$ be the set of nodes of $\PSTrie(T)$.
The set $V$ of nodes of the PLST $\PLST(T)$ for $T$ is a subset of $U$, which is partitioned as $V=V_1 \cup V_2 \subseteq U$.
Nodes in $V_i$ are called \emph{Type~$i$} for $i=1,2$.
The definition of Type~1 and Type~2 nodes follows the one for original LSTs~\cite{crochemore2016linear}.
\begin{enumerate}
	\item A node $u \in U$ is Type~1 if $u$ is a leaf or a branching node in $\PSTrie(T)$.
	\item A node $u \in U$ is Type~2 if $u \notin V_1$ and $\Shrink(u) \in V_1$.
\end{enumerate}
Edges of $\PLST(T)$ are trivially determined: we have $(u, uv) \in V \times V$ as an edge if and only if $v \neq \varepsilon$ and there is no proper nonempty prefix $v'$ of $v$ such that $uv' \in V$.
We will show in Section~\ref{sec:size} that $|V| \in O(|T|)$.
%We say that $u \in V_1 \cup V_2$ is \emph{good} if $\Shrink(u) \in V_1$,
%and $u \in V_1$ is \emph{bad} if $\Shrink(u) \notin V_1\cup V_2$.
%Note that the root $\varepsilon$ is a Type~1 bad node.
We say that $u \in V$ is \emph{good} if $\Shrink(u) \in V$,
and $u \in V$ is \emph{bad} otherwise.
Note that any $u \in V_2$ is good by the definition of $V_2$, and that the root $\varepsilon$ is bad.

To obtain $\str{u}{v}$ for an edge $(u,v)$, if $|v|-|u|=1$, we simply read the edge label $v[1]$ like an LST.
Otherwise, if both $u$ and $v$ are good, we basically use the technique of ``reduction by suffix links''.
An important observation is that the equation $\str{u}{v} = \str{\Shrink(u)}{\Shrink(v)}$, which was a key property to regain the original label in (non-parameterized) LSTs, does not necessarily hold for PLSTs.
Figure~\ref{fig:sign} shows an example,
%\ryoshi{Maybe this example is not needed any more?} 
where $\str{u}{v} = {\tt cb}40 \neq  {\tt cb}00 = \str{\Shrink(u)}{\Shrink(v)}$;
the third symbol $4$ in $\str{u}{v}$ is re-encoded to $0$ in $\str{\Shrink(u)}{\Shrink(v)}$, because the first symbol $v[1]=0$ of $v$, that is referenced by the symbol $4$, is cut out in $\Shrink(v)$.
Fortunately, the possible difference between $\str{u}{v}$ and $\str{\Shrink(u)}{\Shrink(v)}$ is limited.
\begin{observation}\label{obs:reencode}
	Any prev-encoded substring $v$ of text $T$ has at most one position $i$ such that $v[i] = i - 1$. For such a position $i$, we have $\Shrink(v)[i-1] = 0$ and for any $j \in \{2,\dots,|v|\}\setminus\{i\}$, $\Shrink(v)[j-1] = v[j]$. 
	Thus, such a position is unique in $\str{u}{v}$ for each edge $(u, v)$ in $\PLST(T)$.
\end{observation}
For each edge $(u, v)$, we associate an integer named  \emph{re-encoding sign}, so that we can regain $\str{u}{v}$ from $\str{\Shrink(u)}{\Shrink(v)}$ as follows.

\begin{definition}[Re-encoding sign]\label{def:reencoding}
	 For each node $v \in V$, let $u$ be the parent of $v$.
	 We define \emph{re-encoding sign} to $v$ by
	 \[
  		\Recode(v) = \begin{cases}
    	i - |u|           & \mbox{if there exists } i \mbox{ such that } v[i] = i-1 \mbox{ and } |u| < i \leq |v|, \\
    	0           & otherwise.\\
 	    \end{cases}
	\]
\end{definition}
%For instance, $\Recode(0{\tt a}0{\tt b}40) = 3$, see Figure~\ref{fig:sign}~(a).
The re-encoding sign $\Recode(v)$ is well-defined by Observation~\ref{obs:reencode}. 
Figure~\ref{fig:sign} shows an example of re-encoding signs. 
The next lemma immediately follows from Observation~\ref{obs:reencode} and Definition~\ref{def:reencoding}.
\begin{lemma}\label{lem:re-encode} Let $(u, v)$ be an edge in $\PLST(T)$ such that both $u$ and $v$ are good.
	Then for any $i \in \{1,\dots,|\str{u}{v}|\}\setminus\{\Recode(v)\}$, $\str{u}{v}[i]=\str{\Shrink(u)}{\Shrink(v)}[i]$.
	If $\Recode(v) \ge 1$, then $\str{u}{v}[\Recode(v)]=|u|+\Recode(v) -1$ and $\str{\Shrink(u)}{\Shrink(v)}[\Recode(v)] = 0$.
\end{lemma}
Lemma~\ref{lem:re-encode} tells how to recover $\str{u}{v}$ from $\str{\Shrink(u)}{\Shrink(v)}$ using the re-encoding sign at $v$ and the depth $|u|$ of $u$.
Note that the depth $|u|$ is the depth of $u$ in parameterized suffix tries, not the number of nodes from the root to $u$ in PLSTs.

If either $u$ or $v$ is bad in an edge $(u,v)$ with $|v|-|u| \ge 2$, we give up ``reduction by suffix links'' and simply label the edge with the reference to the corresponding substring of the original text $T$, like p-suffix trees.
However, differently from p-suffix trees, not every part of the original text is referenced by an edge in our case.
We keep only the subsequence $T'$ of $T$ obtained by removing parts that are not referred to.
We label an edge $(u,v)$ connected to a bad node with an integer triple $(i, j, k)$ such that $\str{u}{v} = \reencode{T'}_k[i:j]$.

%Our solution is to give $\str{u}{v}$ explicitly on the path from $u$ to $v$ when $u$ is a bad node, without using suffix links.
%We add the following nodes of Type~3.
%\begin{enumerate}  \setcounter{enumi}{2}
%	\item A node $u \in U$ is Type~3 if $u \notin V_1 \cup V_2$ and the parent of $u$ is either a bad Type~1 or  Type~3 node.
%\end{enumerate}
%We will show in Section~\ref{sec:size} that $|V| \in O(|T|)$.
%We say that $u \in V_1 \cup V_2$ is \emph{good} if $\Shrink(u) \in V_1$.
%Otherwise, it is bad, including the case where $u \in V_3$.

%Edges of $\PLST(T)$ are trivially determined: we have $(u, uv) \in V \times V$ as an edge if and only if $v \neq \varepsilon$ and there is no proper nonempty prefix $v'$ of $v$ such that $uv' \in V$.
%The label of the edge $(u, uv)$ is a character or triple. 
%If both $u$ and $uv$ are good or $\str{u}{uv} = 1$, the label of the edge is $v[1]$.
%Otherwise, the edge label is triple $(i, j, k)$, which refers to the pv-string $\reencode{\pre{T}[k:]}[i:j]$.
%$uv$ is called the \emph{$v[1]$-child} of $u$ and denoted by $\Child(u, v[1])$.
%%The label of the edge $(u, uv)$ is $v[1]$ and $uv$ is called the \emph{$v[1]$-child} of $u$ and denoted by $\Child(u, v[1])$.

\begin{figure}[t]
	\centering
	\begin{minipage}[t]{0.485\hsize}
		\centering
		\includegraphics[scale=0.145]{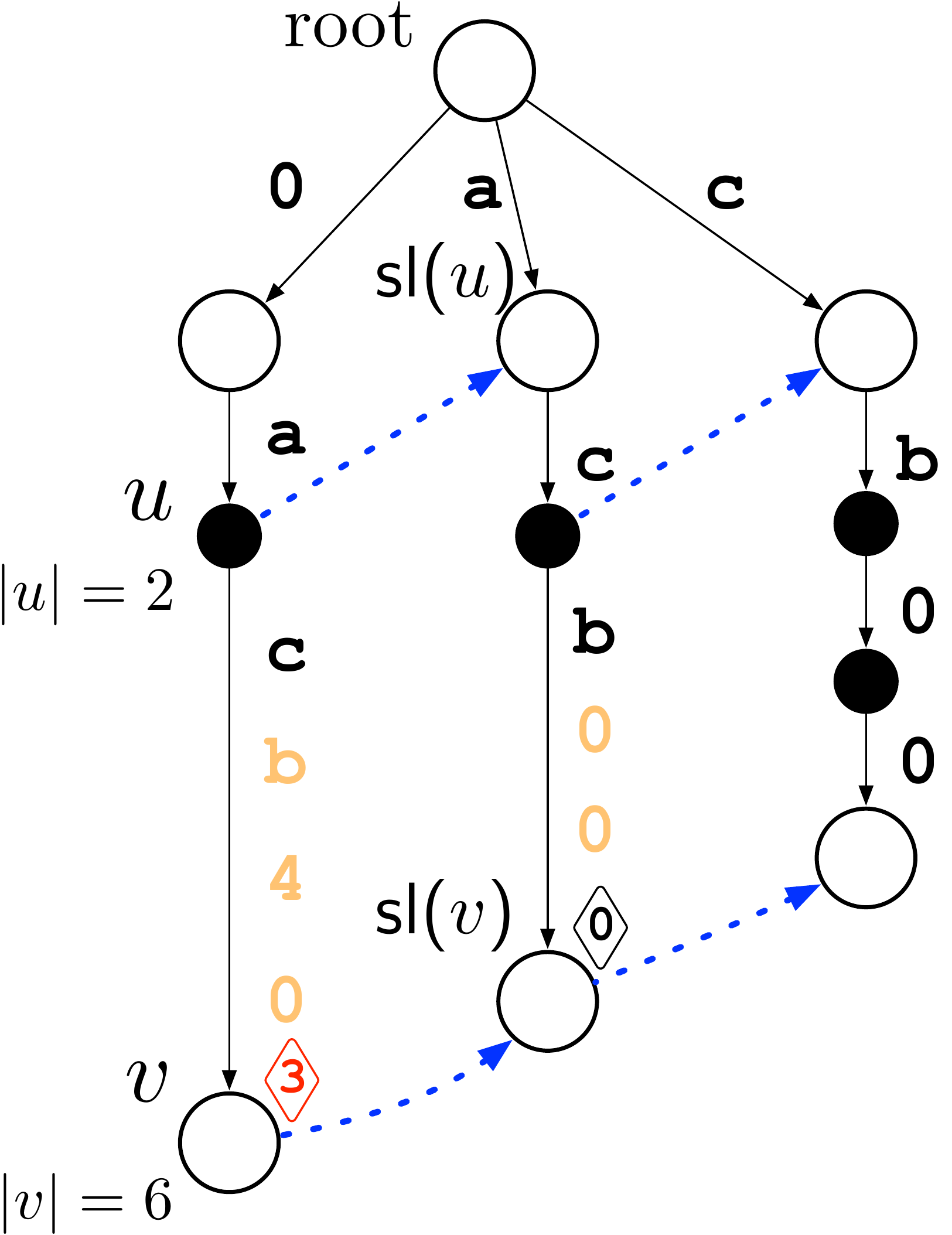}\\
%		\ \ \ \scriptsize{(a)}
		\caption{Illustrating how re-encoding signs are given.  
%		The number in a rhombus represents the re-encoding sign assigned to a node.
		For the node $v = 0{\tt a c b}40$,  we have $\Recode(v) = 5 - 2 = 3$, because $v[5] = 4$ and the parent of $v$ is $u= 0{\tt a}$ of length $2$.
		For the node $\Shrink(v) = {\tt a c b}40$, $\Recode(\Shrink(v)) = 0$.
	}
	\label{fig:sign}
	\end{minipage}
	\hfill
\begin{minipage}[t]{0.485\hsize}
		\centering
		\includegraphics[scale=0.18]{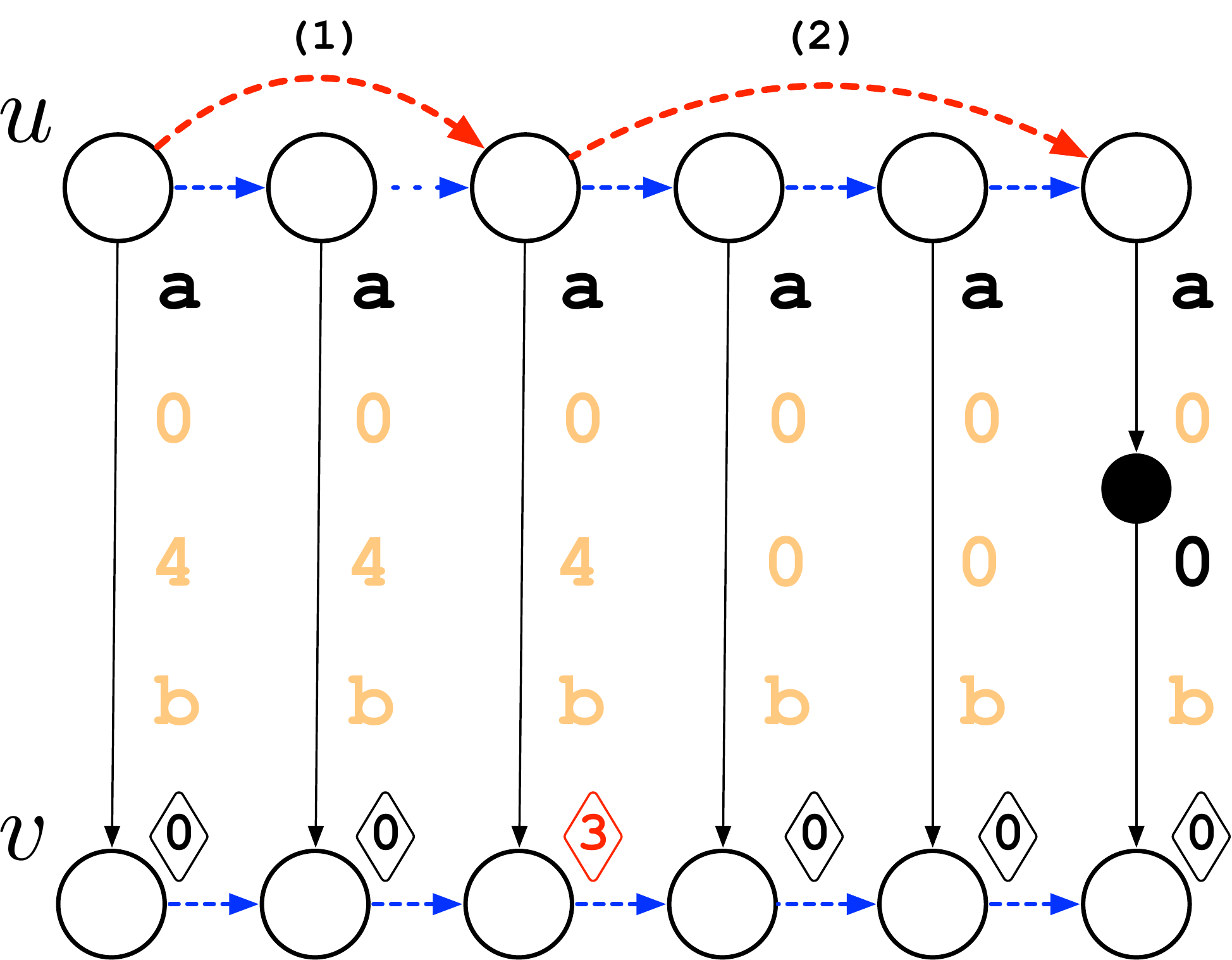}\\
%		\ \ \ \scriptsize{(b)}
		\caption{
		Illustration for our pattern matching algorithm using fast links. 
		In this figure, we check whether $\str{u}{v}$ matches pattern $p = {\tt a}04{\tt b}$ using fast links. 
	}
	\label{fig:fast_link2}
	\end{minipage}
\end{figure}

%As it is, PLSTs are not space-saving data structures then p-suffix trees, thus we reduce retained text.
%In PLSTs, we refer to text only when we recover the label of edge $(u, v)$ such that $u$ or $v$ is a bad node and $\str{u}{v} > 1$.
%Since the intervals of the text referenced is limited, it is sufficient to hold a part of the text.
%Let $T'$ be the subsequence of $\pre{T}$ that consist of substrings of $\pre{T}$ referenced in $\PLST(T)$.
%We have a reference $(i, j, k)$ to $T'$ instead of a reference $(i + l, j + l, k + l)$ to $\pre{T}$ for some $l$.
%It means that we refer to the pv-string $\reencode{T'}_k[i:j]$, where $\reencode{T'}_k[s] = \pre{T[k+l:]}[s+l]$ for $i \leq s \leq j$ and some $l$. 
%Since bad nodes rarely appear in PLSTs, Section~\ref{sec:exp} experimentally  shows that $T'$ is often unnecessary.

In summary, $\PLST(T)$ consists of three kinds of nodes: good Type~1, bad Type~1, and Type~2 (all good).
%In summary, $\PLST(T)$ consists of four kinds of nodes, good Type~1, bad Type~1, Type~2 (all good), and Type~3 (all bad).
If $u \in V$ is a good node, $u$ has its depth, suffix link and re-encoding sign, i.e., the triple $(|u|, \Slink(u), \Recode(u))$, where $\Slink(u)=\Shrink(u)$.
Here we use the notation $\Slink(u)$ to emphasize that the suffix link $\Slink(u)$ is a pointer to the node corresponding to the string $\Shrink(u)$ rather than the string itself.
Therefore, it requires only constant size of memory space.
If $u \in V$ is bad, $u$ dose not have a suffix link, i.e., $u$ has the triple $(|u|, \mathsf{null}, \Recode(u))$.
Each edge $(u, v)$ has either a label character or triple;
if both $u$ and $v$ are good or $\str{u}{v} = 1$, the edge label is $\str{u}{v}[1]$. 
Otherwise, the edge label is a triple $(i, j, k)$ such that $\str{u}{v} = \reencode{T'}_k[i:j]$. 
If some bad nodes appear in $\PLST(T)$, we need the subsequence $T'$ of $\pre{T}$ to recover the labels of edges connecting the bad nodes. 
Otherwise, we do not need any text.

We remark that another idea to overcome the problem of the absence of $\Shrink(u)$ in a PLST for a node $u$ might be to add $\Shrink^i(u)$ to $V$ for all $i =1,\dots,|u|$ so that $V$ is closed under $\Shrink$, where $\Shrink^i(u) = \Shrink(\Shrink^{i-1}(u))$ and $\Shrink^0(u) = u $.
However, there exists a series of texts
$T_n = \tx_1 \ta_1 \dots \tx_n \ta_n \tx_1 \ta_1 \dots \tx_n \ta_n \ty_1 \ta_1 \dots \ty_n \ta_n \tz \$$
where $\tx_i, \ty_i, \tz \in \Pi$ and $\ta_i \in \Sigma$ for each $i$, for which the number of those additional nodes will be $\Omega(|T_n|^2)$. Thus, the size of the index structures cannot be kept in linear.
% as we show in Appendix~\ref{app:naive}.
%Instead, when $u$ or $v$ is a bad node and $\str{u}{v} > 1$, we refer to the text like p-suffix tree to get $\str{u}{v}$. 
%$T'$ is the subsequence of $\pre{T}$ that consist of substrings of $\pre{T}$ as referenced by edge labels including bad nodes.

%%%%%%%%%%%%%%%%%%%%%%%%%%%%%%%%%%%%%%%%%%%%%%%%%%%%%%%%%%%%%%%%%%%%%%%%%%%%%%%%%%%%%%%%%%%%%%%%%%%%%%%%%%%%%%%%%%%%%%%%%%%%
\subsection{Parameterized pattern matching with PLSTs}
This subsection presents our algorithm for solving the parameterized pattern matching problem as an application of PLSTs.
The function $\textsc{P-Match}$ of Algorithm~\ref{alg:simplematch} takes a prev-encoded string $p$ and a node in $\PLST(T)$
 and checks whether there is $v \in \PSub(T)$ such that $p = \str{u}{v}$.
 If it is the case, it returns the least extension ${v'}$ of $v$  such that $\vtx{v'} \in V$.
 In other words, $p$ is a prefix of $\str{u}{v'}$, where $v'$ should be $v$ itself if $v \in V$.
 Otherwise, it returns $\mathsf{null}$.
%If a p-string pattern $P$ p-matches substrings of $T$ at $k$ positions $j_1,\dots,j_k$, then $v = \textsc{P-Match}(\pre{P},\vtx{\varepsilon})$ will be a node whose descendant leaves are exactly $\vtx{\pre{T[j_1:]}},\dots,\vtx{\pre{T[j_k:]}}$.

For an input pair $(p, u)$, if $p = \varepsilon$, then $\textsc{P-Match}$ returns $u$, as it is required.
Otherwise, it first tries to regain $\str{u}{v}$ for the $p[1]$-child $\vtx{v}$ of $\vtx{u}$, if $\vtx{u}$ has such a child.
At first, suppose $|p| \ge |v|-|u| = l$.
We would like to know whether $p[1:l] = \str{u}{v}$.
If $l=1$, it means that we have already confirmed that $p[1:l] = \str{u}{v}$.
Then we just go down to $v$ and recursively call $\textsc{P-Match}(p[2:],v)$.
If $l \ge 2$ and either $u$ or $v$ is bad, we refer to $T'$ and check if $p[1:l] = \str{u}{v}$ as with matching in a p-suffix tree.
If $l \ge 2$ and both $u$ and $v$ are good, we cannot know from the edge $(u, v)$ itself what $\str{u}{v}$ is except for its first symbol $\str{u}{v}[1] = p[1]$.
To recover whole $\str{u}{v}$, we use the suffix link of $u$.
Since $u$ is good, $\Slink(u)$ is defined.
%Since $|v|-|u| \ge 2$, $u$ is a good node by Observation~\ref{obs:bad Type3}, and thus $\Slink(u)$ is defined.
If $\Recode(v) = 0$, we have $\str{u}{v}=\str{\Shrink(u)}{\Shrink(v)}$ by Lemma~\ref{lem:re-encode},
and we simply call $\textsc{P-Match}(p[1:l],\Slink(u))$.
 Otherwise, %by Lemma~\ref{lem:re-encode},
  we have $p[1:l] = \str{u}{v}$ if and only if $p[\Recode(v)] = |u| + \Recode(v) -1$ and $p'[1:l] = \str{\Shrink(u)}{\Shrink(v)}$, where 
\[
	p'[i] = \begin{cases} 0 & \text{if $i = \Recode(v)$,}
	\\	p[i] & \text{otherwise} \end{cases}
\]
for $i=1,\dots,|p|$.
Thus, the recursive call of $\textsc{P-Match}(p'[1:l],\Slink(u))$ returns $\mathsf{null}$ iff $p[1:l] \neq \str{u}{v}$.
%We note that it may be the case that $\Shrink(v) \notin V$, but this does not matter for our algorithm.
%The recursive call checks whether $p'[1:l] = \str{\Shrink(u)}{\Shrink(v)}$ but $\Shrink(v)$ is not an argument and not used.
If $\textsc{P-Match}(p'[1:l],\Slink(u))$ returns a node, then $p[1:l] = \str{u}{v}$ and thus we continue matching by calling $\textsc{P-Match}(p[l+1:],v)$.

The above discussion is valid when $|p| \le |v|-|u|$.
If $\Recode(v) =0$ or $\Recode(v) > |p|$, then $p$ is a prefix of $\str{u}{v}$ iff $p$ is a prefix of $\str{\Shrink(u)}{\Shrink(v)}$.
Otherwise, $p$ is a prefix of $\str{u}{v}$ iff $p[\Recode(v)] = |u|+\Recode(v)-1$ and $p'$ is a prefix of $\str{\Shrink(u)}{\Shrink(v)}$.
Thus the recursion is justified.
If $\textsc{P-Match}(p'[1:l],\Slink(u))$ returns a node, $p$ is a prefix of $\str{u}{v}$ and we call $\textsc{P-Match}(\varepsilon,v)$, which returns $v$.

\begin{algorithm2e}[t]
%\caption{Parameterized pattern matching algorithm ($\textsc{P-Match}(p, u)$)}
\caption{$\textsc{P-Match}(p, u)$}
\label{alg:simplematch}
\KwIn{A string $p$ and a node $u$ in $\PLST(T)$}
\KwOut{The highest descendant $v$ of $u$ such that $p$ is a prefix of $\str{u}{v}$}
\lIf{$p = \varepsilon$}{%
	\Return $u$}
\Else{
	\lIf{$\Child(u, p[1])$ is undefined}{%
		\Return $\mathsf{null}$}
	\Else{
		$v \leftarrow \Child(u, p[1])$\;
		$l \leftarrow \min\{|p|,|v|-|u|\}$\;
		\If{$l \geq 2$ and $u$ or $v$ is bad}{
				let $\alpha$ be the label of the edge $(u, v)$\;
				\lIf{$p[:l] \neq \alpha[:l]$}{%
					\Return $\mathsf{null}$}
		}
		\ElseIf{$l \ge 2$}{
			\If{$1 \le \Recode(v) \le |p|$}{
				\lIf{$p[\Recode(v)] = |u| + \Recode(v) -1$}{%
					$p[\Recode(v)] \leftarrow 0$%
				}\lElse{%
					\Return $\mathsf{null}$}
			}
			\lIf{$\textsc{P-Match}(p[1:l],\Slink(u)) = \mathsf{null}$\label{ln:fastlink}}{%
				\Return $\mathsf{null}$}
		}
	}
	\Return $\textsc{P-Match}(p[l+1:],v)$\;
}
\end{algorithm2e}

\begin{proposition}\label{lem:simple}
We can decide whether $T$ has a substring that p-matches $P$ using Algorithm~\ref{alg:simplematch}.
\end{proposition}

%\begin{figure}[t]
%	\centering
%	\begin{minipage}[t]{0.47\hsize}
%		\centering
%		\includegraphics[scale=0.26]{fig/algo2}\\
%%		\ \ \ \scriptsize{(a)}
%		\caption{
%		The fast link (red broken arrow) is obtained by skipping intermediate nodes visited by the suffix links (blue broken arrows).
%%		We jump from $u$ to $\Slink^3(u)$ by using the fast link $\Flink(u, v)$.
%		}
%		\label{fig:fast_link}
%	\end{minipage}
%	\hfill
%	\begin{minipage}[t]{0.47\hsize}
%		\centering
%		\includegraphics[scale=0.24]{fig/fastlink_match3}\\
%%		\ \ \ \scriptsize{(b)}
%		\caption{
%		Illustration for our pattern matching algorithm using fast links. 
%		In this figure, we check whether $\str{u}{v}$ matches pattern $p = {\tt a}04{\tt b}$ and use fast links 4 times. 
%	}
%	\label{fig:fast_link2}
%	\end{minipage}
%\end{figure}

The time complexity of Algorithm~\ref{alg:simplematch} is not linear as it is.
Suppose that $\textsc{P-Match}(p, u)$ is called.
It can be the case $|v|-|u| \ge |p| \ge 2$ and either $\Recode(v) = 0$ or $\Recode(v) > l$ where $v = \Child(u, p[1])$.
In this case, the algorithm simply calls $\textsc{P-Match}(p,\Slink(u))$, where the first argument has not changed from the preceding call.
Such recursion may be repeated, and amortized time complexity is not linear.
%For example, when $\textsc{P-Match}(\mathtt{b}49\mathtt{a}3\texttt{\$},u)$ is called for the PLST in Figure~\ref{fig:fast_link},
%it recursively calls $\textsc{P-Match}(\mathtt{b}49\mathtt{a}0\texttt{\$},\Slink(u))$, which calls $\textsc{P-Match}(\mathtt{b}49\mathtt{a}0\texttt{\$},\Slink^2(u))$,
%which again calls $\textsc{P-Match}(\mathtt{b}49\mathtt{a}0\texttt{\$},\Slink^3(u))$.
% shows an example of such an edge $(u, v)$ where $u=00\mathtt{ab}00\mathtt{a}$.
The same difficulty and a solution have already been discussed by Crochemore et al.~\cite{crochemore2016linear} for
% (non-parameterized) LSTs.
LSTs.
Following them, we introduce \emph{fast links} as follows, which allow us to skip recursions that always preserve the first argument.

\begin{definition}[Fast link]
\label{def:fastlink}
	For each edge $(u, v) \in V \times V$ such that $|v| - |u| > 1$ and both $u$, $v$ are good, the \emph{fast link} for $(u, v)$ is defined to be $\Flink(u, v) = \Slink^{k}(u)$ where $k \ge 1$ is the smallest integer satisfying either
	$|v_k| < |v| - k$ or $0 < \Recode(v_k)$,
	where $v_k = \Child(\Slink^k(u),a)$ for $a = \str{u}{v}[1]$.
\end{definition}

%\begin{definition}[Fast link]
%\label{def:fastlink}
%	For each edge $(u, v) \in V \times V$ such that $|v| - |u| > 1$, the \emph{fast link} for $(u, v)$ is defined to be $\Flink(u, v) = \Slink^{k}(u)$ where $k \ge 1$ is the smallest integer satisfying either
%%	\begin{enumerate}
%%		\item $|v_k| < |v| - k$, or,
%%		\item $0 < \Recode(v_k) < |v| - k$,
%%	\end{enumerate}
%	$|v_k| < |v| - k$ or $0 < \Recode(v_k) < |v| - k$,
%	where $v_k = \Child(\Slink^k(u),a)$ for $a = \str{u}{v}[1]$.
%\end{definition}

%\begin{observation}\label{obs:reencode}
%	Let $(u, v)$ be an edge in $\PLST(T)$ such that $u$ is not the root.
%	satisfying either
%	\begin{enumerate}
%		\item $|\str{u}{v}| > |\str{u'}{v'}|$, or,
%		\item for any $i \in \{1,\dots,|\str{u}{v'}|\}\setminus\{\Recode(v')\}$, $\str{u'}{v'}[i] = \str{\Shrink(u')}{\Shrink(v')}[i]$, $\str{u}{v}[\Recode(v')]=|u|+\Recode(v') -1$ and $\str{\Shrink(u')}{\Shrink(v')}[\Recode(v')] = 0$
%	\end{enumerate}
%	where $u' = \Flink(u, v)$ and $v' = \Child(\Flink(u, v),a)$ for $a = \str{u}{v}[1]$.
%\end{observation}

%\begin{restatable}{lemma}{RecoverEdge}
%\label{lem:RecoverEdge}
%We can decide whether the underlying label of a gibven edge $(u,v)$ of length $l$ p-matches $p$ of length $l$ by using fast links.
%\end{restatable}

Algorithm~\ref{alg:simplematch} will run in linear time by replacing $\Slink(u)$ in Line~\ref{ln:fastlink} by $\Flink(u, v)$.
If $|v_k| < |v| - k = |\Shrink^k(v)|$, the node $v_k$ occurs between $\Shrink^k(u)$ and $\Shrink^k(v)$.
Then, $\textsc{P-Match}(p,\Slink^k(u))$ will call $\textsc{P-Match}(p[1:|v_k|-|\Slink^k(u)|],\Slink^{k+1}(u))$. 
%unless $|p| \le |v_k|-|\Slink^k(u)|$.
%  $|v_k|-|\Slink^{k}(u)| < |p|$.
When $0 < \Recode(v_k)$, we change the $\Recode(v_k)$-th symbol of $p$, which must be a positive integer, to $0$.
Therefore, the number of fast links we follow is bounded by $2|p|$. 
Figure~\ref{fig:fast_link2} shows how to p-match $\str{u}{v}$ and $p = {\tt a}04{\tt b}$ using fast links. 
We know that $p[1] = \str{u}{v}[1] = {\tt a}$.
After following the fast link (1), we check whether $p[\Recode(\Shrink^2(v))]=4$ and rewrite the value of $p[\Recode(\Shrink^2(v))]$ to $0$.
After using (2), we check whether $p[3]= 0$. 
In this way, we can know that $p$ matches $\str{u}{v}$.

%One might expect that fast links would enable us to solve the matching problem in linear time,
%since we can always find either a node between $\Flink(u,v)=\Slink^i(u)$ and $\Slink^i(v)$, or a re-encoding sign zeroizing a positive integer between $\Slink^i(u)$ and $\Slink^i(v)$.
%However, this is not always the case for $\Flink(\Flink(u,v),v_i)=\Slink^{i+j}(u)$ and $\Slink^{i+j}(v)$, where $v_i$ is defined as in Definition~\ref{def:fastlink}.
%It can be the case that $v_i$ is much bigger than $\Slink^i(v)$ and
% there is neither a node or a position referred to by a re-encoding sign between $\Flink(\Flink(u,v),v_i)=\Slink^{i+j}(u)$.
%% though it is guaranteed that either there must be a node or a position referred to by a re-encoding sign between $\Slink^j(\Flink(u,v))$ and $\Slink^j(v_i)$.
%We found neither a proof or a counterexample for the expectation for linear-time performance of Algorithm~\ref{alg:simplematch} modified with fast links.
%This is future work.

%We relegate a proof of Theorem~\ref{lem:p-match} to Appendix~\ref{app:matching}.

\begin{restatable}{theorem}{pmatch}
\label{lem:p-match}
Given $\PLST(T)$ and a pattern $P$ of length $m$, we can decide whether $T$ has a substring that p-matches $P$ in $O(m)$ time.
\end{restatable}

\subsection{The size of PLSTs}\label{sec:size}

We now show that the size of $\PLST(T)$ is linear with respect to the length $n$ of a text $T$.
First, we show a linear upper bound on the number of nodes of $\PLST(T)$.
The nodes of Type~1 appear in the p-suffix tree, so they are at most $2n$~\cite{PMA}.
It is enough to show that the number of nodes of Type~2 is linearly bounded as well.
%We relegate proofs of lemmas to Appendices.

%\begin{lemma}
\begin{restatable}{lemma}{TypeTwo}
\label{lem:Type2}
The number of Type~2 nodes in $\PLST(T)$ is smaller than $2n$.
\end{restatable}
%\end{lemma}

\begin{proof}
	Let us consider an implicit suffix link chain in $\PSTrie(T)$ starting from $w = \pre{T[:k]}$ with $1 \le k < n$, i.e.,
	 $(w,\Shrink(w),\Shrink^2(w),\dots,\Shrink^{|w|}(w))$.
	 $\PSTrie(T)$ has $n-1$ such chains and every internal node of $\PLST(T)$ appears in at least one chain.
	 If a chain has two distinct Type~2 nodes $\Shrink^i(w)$ and $\Shrink^j(w)$ with $i < j$, since $\Shrink^{i+1}(w)$ is Type~1 by definition,
	 one can always find a Type~1 node between them.

	Define a binary relation $R$ between $V_1$ and $V_2$ by
	\[
		R = \{\, (u,v) \in V_1 \times V_2 \mid \text{there is $i$ s.t.\ $v = \Shrink^i(u)$ and $\Shrink^j(u) \notin V_2$ for all $j < i$}\,\}
	\,\]
	and let $R_2 = \{\, v \in V_2 \mid (u,v) \in R \text{ for some } u \in V_1\,\}$.
	Since $R$ is a partial function from branching nodes to  Type~2 nodes, we have $|R_2| \le n$.
	By the above argument on a chain, each chain has at most one Type~2 node $v \in V_2$ such that $v \notin R_2$.
	Since there are $n-1$ chains, we have $|V_2 \setminus R_2| < n$.
	All in all, $|V_2| = |R_2|+|V_2 \setminus R_2| < n + n = 2n$.
	\qed
\end{proof}

The number of edges and their labels, as well as 
the number of suffix links, depth and re-encoding sign for nodes, is asymptotically bounded above by the number of nodes in $\PLST(T)$. 
$T'$ is a subsequence of $\pre{T}$, thus its length is $O(n)$.  
Therefore, the size of $\PLST(T)$ is $O(n)$.

\begin{theorem}
 Given a p-string $T$ of length $n$, the size of $\PLST(T)$ is $O(n)$.% in total.
\end{theorem}

\section{Experiments}\label{sec:exp}

We performed comparative experiments on the number of nodes of PLSTs and p-suffix trees for four sorts of text strings changing their length.
Text strings we used are
 random strings over a constant alphabet $\Sigma$ with $|\Sigma|=2$ and those over a parameter alphabet $\Pi$ with $|\Pi|=2$, 
 and
 Fibonacci strings over $\Sigma$ with $|\Sigma|=2$ and those over $\Pi$ with $|\Pi|=2$.
PLSTs for constant strings are of course identical to LSTs. 
%We experimented while changing text length $n$ for four types of strings with a sentinel symbol $\$$ at the end.
%Two types of random strings of the alphabet size $|\Sigma| = 2$, $|\Pi|= 0$ or $|\Sigma| = 0$, $|\Pi|= 2$.
%Two types of Fibonacci strings of the alphabet size $|\Sigma| = 2$, $|\Pi|= 0$ or $|\Sigma| = 0$, $|\Pi|= 2$.
%綺麗にまとめたい
%it ends with a sentinel symbol $\$ \in \Sigma$.
%We measured the number of nodes on each experiment.
For random strings, we measured the average number of nodes for 100 strings of each length $n=10,\dots, 10240$.
For Fibonacci strings, we measured the number of nodes for each of the 11th through 22nd Fibonacci strings.
% with a sentinel symbol $\$$ at the end.
%Moreover, we determined the word size required for the data structure from the calculated number of nodes.
%Here, a string of length $n$ is $n$ words, and each value is 1 word.
%Therefore, PLST and p-suffix tree require 4 words per node.
%Experiments are executed on a machine with Intel Xeon CPU E5-2609 8 cores 2.40 GHz, 256 GB memory, and Debian Wheezy operating system.
The results of our experiments are shown in Table~\ref{tab:exp}. %Tables~\ref{tab:random} and \ref{tab:fibonacci}.
Recall that p-suffix trees consist of Type~1 nodes, while PLSTs have Type~2 nodes in addition.
For random strings, we can see that the number of Type~2 nodes is close to the text length.
On the other hand, for Fibonacci strings, PLSTs have few Type~2 nodes.
In these experiments, since no bad node appeared except the root, PLSTs did not need any text, that is, $T' = \epsilon$.

%P-suffix trees by Lee et al.~\cite{lee2011line} keep the pv-string $\pre{T}$ and each edge has the first character of the edge label and a triple which is a reference to some interval of the text.
%PLSTs keep the subsequence $T'$ of $\pre{T}$ and each node has a triple of its depth, suffix link, and re-encoding sign and has the first character of the edge label.
Because the size of each node is the same in a p-suffix tree and a PLST, the difference of the memory efficiency of the two data structures is just the difference of the memory size for $\pre{T}$ and the Type~2 nodes (and $T'$ if necessary).
The experimental results suggest that PLSTs use less memory than p-suffix trees for indexing highly repetitive strings such as Fibonacci strings.
%We conducted more experiments with other types of strings. See Appendix~\ref{app:expe} for those experiments.

\begin{table}[t]
%\begin{minipage}[t]{.485\textwidth}
\begin{center}
\caption{The numbers of nodes of PLSTs for different sorts of strings}
\scalebox{0.8}{
\label{tab:exp}
  \begin{tabular}{|r||p{4em}|p{3.5em}||p{4em}|p{3.5em}|} \hline
  	&	\multicolumn{4}{c|}{random strings}
\\ \cline{2-5}
      & \multicolumn{2}{c||}{constant string} & \multicolumn{2}{|c|}{p-string} \\ \hline
     \multicolumn{1}{|c||}{\small{length $n$}} & \multicolumn{1}{|c|}{Type~1} & \multicolumn{1}{c||}{Type~2} & \multicolumn{1}{|c|}{Type~1} & \multicolumn{1}{c|}{Type~2} \\\hline  \multicolumn{5}{|c|}{}  \\[-10pt] \hline
   10 &\hfill 16.98 &\hfill 6.04 &\hfill 16.93 &\hfill 5.23\\ 
   20 &\hfill 35.66 &\hfill 12.78 &\hfill 35.72 &\hfill 12.27 \\
   40 &\hfill 74.58 &\hfill 27.25 &\hfill 74.53 &\hfill 26.22\\
   80 &\hfill 153.61 &\hfill 56.82 &\hfill 153.48 &\hfill 56.04\\ 
  160 &\hfill 312.37 &\hfill 115.55 &\hfill 312.45 &\hfill 115.24\\ \hline
  320 &\hfill 631.40 &\hfill 234.55 &\hfill 631.27 &\hfill 235.32\\
  640 &\hfill 1270.34 &\hfill 477.29 &\hfill 1270.47 &\hfill 475.34\\ 
 1280 &\hfill 2549.35 &\hfill 956.18 &\hfill 2549.39 &\hfill 957.03\\
 2560 &\hfill 5108.37 &\hfill 1923.62 &\hfill 5108.48 &\hfill 1922.97\\
 5120 &\hfill 10227.48 &\hfill 3845.35 &\hfill 10227.29 &\hfill 3853.97\\ \hline
10240 &\hfill 20466.49 &\hfill 7710.50 &\hfill 20466.14 &\hfill 7704.25\\  \hline
%20480 &\hfill 0 &\hfill 0 &\hfill 0 &\hfill 0\\ \hline
    \end{tabular}
    }
  \quad
\scalebox{0.8}{
%\label{tab:fibonacci}
  \begin{tabular}{|r||p{3.5em}|p{3.5em}||p{3.5em}|p{3.5em}|} \hline
  	&	\multicolumn{4}{c|}{Fibonacci strings}
\\ \cline{2-5}
      & \multicolumn{2}{c||}{constant string} & \multicolumn{2}{|c|}{p-string} \\ \hline
     \multicolumn{1}{|c||}{length $n$} & \multicolumn{1}{|c|}{Type~1} & \multicolumn{1}{c||}{Type~2} & \multicolumn{1}{|c|}{Type~1} & \multicolumn{1}{c|}{Type~2} \\
	\hline  \multicolumn{5}{|c|}{}  \\[-10pt] \hline
%   35 &\hfill 68 &\hfill 9 &\hfill 67 &\hfill 9\\ 
%   56 &\hfill 107 &\hfill 9 &\hfill 107 &\hfill 10 \\
   90 &\hfill 178 &\hfill 12 &\hfill 177 &\hfill 12\\
  145 &\hfill 285 &\hfill 12 &\hfill 285 &\hfill 13\\ 
  234 &\hfill 466 &\hfill 15 &\hfill 465 &\hfill 15\\ 
  378 &\hfill 751 &\hfill 15 &\hfill 751 &\hfill 16\\
  611 &\hfill 1220 &\hfill 18 &\hfill 1219 &\hfill 18\\ \hline 
  988 &\hfill 1971 &\hfill 18 &\hfill 1971 &\hfill 19\\
 1598 &\hfill 3194 &\hfill 21 &\hfill 3193 &\hfill 21\\
 2585 &\hfill 5165 &\hfill 21 &\hfill 5165 &\hfill 22\\
 4182 &\hfill 8362 &\hfill 24 &\hfill 8361 &\hfill 24\\ 
 6766 &\hfill 13527 &\hfill 24 &\hfill 13552 &\hfill 25\\ \hline
 10947 &\hfill 21892 &\hfill 27 &\hfill 21918 &\hfill 27\\ \hline
% 17712 &\hfill 35419 &\hfill 27 &\hfill 0 &\hfill 28\\ \hline
    \end{tabular}
    }
  \end{center}
\end{table}

\section{Conclusion and future work}\label{sec:conclusion}
%This paper aims at designing an indexing structure for parameterized strings.
%The proposed structure has size $O(|T|)$ for a text $\para{T}$.
%We experimentally showed that the proposed structure is space-saving compared with parameterized suffix trees for indexing highly repetitive strings such as Fibonacci strings.
%A useful indexing structure must accompany a linear-time matching algorithm, but we have not managed to find such an algorithm.
%Moreover, for various applications, like computing the longest common substrings, an efficient construction is also required.
%Those are left for future work.
%We hope these will be solved positively and the ideas would be useful to generalize CDAWGs~\cite{crochemore1997direct,blumer1987complete} and L-CDAWGs~\cite{takagi2017linear} to a data structure for parameterized strings. 

In this paper, we presented an indexing structure called a PLST for the parameterized pattern matching problem.
%The size of the PLST for $\para{T}$ is $O(n)$. 
Given a p-string $T$ of length $n$, the size of PLST for $\para{T}$ is $O(n)$.
We presented an algorithm that solves the problem in $O(m)$ time, where $m$ is the length of the pattern.
We experimentally showed that PLST is space-saving from p-suffix tree for indexing highly repetitive strings such as Fibonacci strings.

For PLSTs to be useful for various applications, like computing the longest common substrings, an efficient algorithm for constructing PLSTs is required like LSTs~\cite{onlineLST}.
%Also, there is no theoretical guarantee that Algorithm~\ref{alg:simplematch} will work in linear time, even with fast links.
%Analysis of the operation time of Algorithm~\ref{alg:simplematch} when using fast link is a future work.
Furthermore, the ideas developed in this paper may be useful to generalize 
%CDAWGs~\cite{crochemore1997direct,blumer1987complete} and 
L-CDAWGs~\cite{takagi2017linear} to a data structure for parameterized strings.

\bibliographystyle{plainurl} 
\newpage
\bibliography{docs/ref}
\newpage
\appendix
\section{Appendix}

\subsection{The implicit suffix link closure of branching nodes is too big}\label{app:naive}
{
\begin{figure}[h]
	\centering
	\begin{minipage}{\hsize}
		\centering
		\includegraphics[scale=0.14]{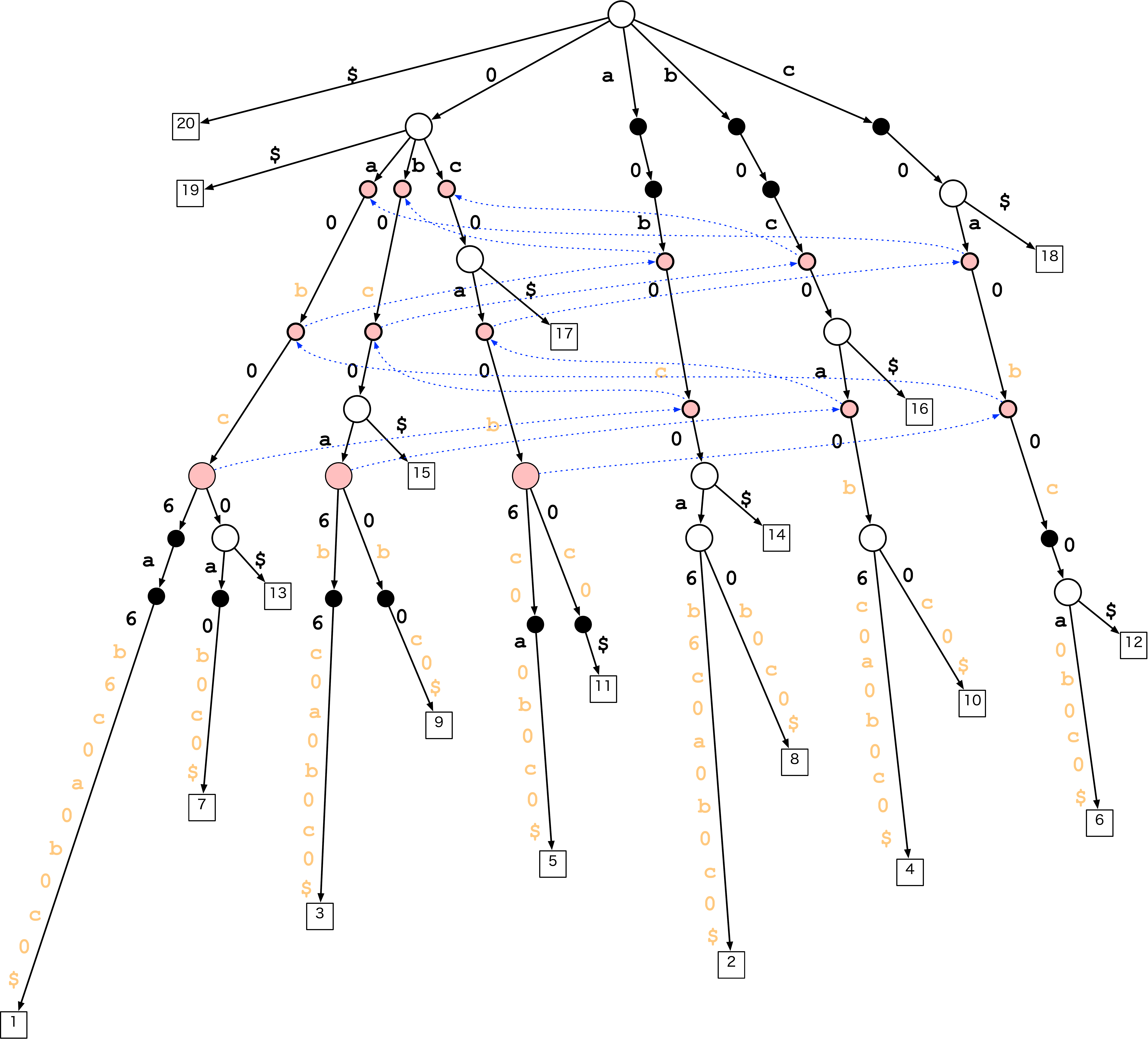}\\
	\end{minipage}
	\caption{
		An example demonstrating that the implicit suffix link closure $\{\, \Shrink^{i}(u) \in \PSTrie(T) \mid u \in V_1 \text{ and } i \le |u| \,\}$ of Type~1 nodes has too many elements,
		where $T={\tt taubvctaubvcwaxbycz\$}$ ($\pre{T} = {\tt \z a\z b\z c\s a\s b\s c\z a\z b\z c\z\$}$) with $\Sigma=\{ {\tt a, b, c,\$}\}$ and $\Pi=\{{\tt t, u, v, w, x, y, z}\}$. 
		Big and small red circles represent bad nodes in $V_1$ and newly added nodes not in $V_1 \cup V_2$, respectively.
	}
	\label{fig:counter_example_type3}
\end{figure}

We show that the total number of nodes of the form $\Shrink^j(u) \in \PSub(T)$ for some $u \in V_1$ cannot be linearly bounded by $|T|$.
Let us consider a text
\[
T_n = \tx_1 \ta_1 \dots \tx_n \ta_n \tx_1 \ta_1 \dots \tx_n \ta_n \ty_1 \ta_1 \dots \ty_n \ta_n \tz \$
\]
where $\tx_i, \ty_i, \tz \in \Pi$ and $\ta_i \in \Sigma$ for each $i$.
Note that $|T_n| \in O(n)$.
Here
\[
	w_i = 0 \ta_i 0 \ta_{i+1} \dots 0 \ta_n 0 \ta_1 \dots 0 \ta_{i-1} \in \PSub(T_n)
\]
 is a Type~1 node, since $ w_i 0 ,  w_i (2n) \in \PSub(T_n)$.
Then the set $\{\, \Shrink^j(w_i) \mid 1 \le i \le n,\,0 \le j < 2n\,\}$ has $2n^2$ elements.
Therefore, we cannot keep our indexing structure in linear size. 
Figure~\ref{fig:counter_example_type3} illustrates the case of $n=3$, where twelve additional nodes are created.
}

\subsection{Other experiments}\label{app:expe}
We performed comparative experiments on the numbers of nodes of PLSTs and p-suffix trees for texts in addition to random and Fibonacci strings.
The results of our experiments for Thue-Morse strings and Period-doubling strings are shown in Tables~\ref{tab:exp2}.
For Thue-Morse strings and Period-doubling strings, our data structure only have a limited number of additional nodes.
The Fibonacci strings, Thue-Morse strings and Period-doubling strings are defined as follows.

The $k$-th Fibonacci string $\mi{Fib}_k$ is defined by the following recurrence:
%\ryoshi{0-th ?}
\begin{align*}
%	\mi{Fib}_0= \epsilon,\ 
	\mi{Fib}_1 = {\tt b},\ \mi{Fib}_{2} = {\tt a},\
	\mi{Fib}_k = \mi{Fib}_{k-1} + \mi{Fib}_{k-2}  \text{ for $k > 2$}\,.
\end{align*}

The $k$-th Thue-Morse string can be obtained by applying the following homomorphism $\sigma$ to $\tt{a}$ $k$ times:
\begin{align*}
	\sigma(\tt{a}) = \tt{ab} \\
	\sigma(\tt{b}) = \tt{ba} 
\end{align*}

The $k$-th Period-doubling string can be obtained by applying the following homomorphism $\sigma$ to $\tt{a}$ $k$ times:
\begin{align*}
	\sigma(\tt{a}) = \tt{ab} \\
	\sigma(\tt{b}) = \tt{aa} 
\end{align*}

\begin{table}[t]
\begin{center}
\caption{The numbers of nodes of PLSTs for Thue-Morse and Period-doubling strings}
\scalebox{0.82}{
\label{tab:exp2}
  \begin{tabular}{|r||p{4em}|p{3.5em}||p{4em}|p{3.5em}|} \hline
  	&	\multicolumn{4}{c|}{Thue-Morse strings}
\\ \cline{2-5}
      & \multicolumn{2}{c||}{constant string} & \multicolumn{2}{|c|}{p-string} \\ \hline
     \multicolumn{1}{|c||}{\small{length $n$}} & \multicolumn{1}{|c|}{Type~1} & \multicolumn{1}{c||}{Type~2,3} & \multicolumn{1}{|c|}{Type~1} & \multicolumn{1}{c|}{Type~2,3} \\\hline  \multicolumn{5}{|c|}{}  \\[-10pt] \hline
   17 &\hfill 28 &\hfill 10 &\hfill 29 &\hfill 6\\
   33 &\hfill 56 &\hfill 14 &\hfill 57 &\hfill 8\\
   65 &\hfill 112 &\hfill 18 &\hfill 113 &\hfill 10\\ 
  129 &\hfill 224 &\hfill 22 &\hfill 225 &\hfill 12\\ 
  257 &\hfill 448 &\hfill 26 &\hfill 449 &\hfill 14\\ \hline
  513 &\hfill 896 &\hfill 30 &\hfill 897 &\hfill 16\\ 
 1025 &\hfill 1792 &\hfill 34 &\hfill 1793 &\hfill 18\\
 2049 &\hfill 3584 &\hfill 38 &\hfill 3585 &\hfill 20\\
 4097 &\hfill 7168 &\hfill 42 &\hfill 7169 &\hfill 22\\ 
 8193 &\hfill 14336 &\hfill 46 &\hfill 14337 &\hfill 24\\  \hline
    \end{tabular}
    }
  \quad
\scalebox{0.82}{
  \begin{tabular}{|r||p{3.5em}|p{3.5em}||p{3.5em}|p{3.5em}|} \hline
  	&	\multicolumn{4}{c|}{Period-doubling strings}
\\ \cline{2-5}
      & \multicolumn{2}{c||}{constant string} & \multicolumn{2}{|c|}{p-string} \\ \hline
     \multicolumn{1}{|c||}{length $n$} & \multicolumn{1}{|c|}{Type~1} & \multicolumn{1}{c||}{Type~2,3} & \multicolumn{1}{|c|}{Type~1} & \multicolumn{1}{c|}{Type~2,3} \\
	\hline  \multicolumn{5}{|c|}{}  \\[-10pt] \hline
%   35 &\hfill 68 &\hfill 9 &\hfill 67 &\hfill 9\\ 
%   56 &\hfill 107 &\hfill 9 &\hfill 107 &\hfill 10 \\
   17 &\hfill 30 &\hfill 7 &\hfill 31 &\hfill 9\\
   33 &\hfill 64 &\hfill 11 &\hfill 61 &\hfill 11\\ 
   65 &\hfill 126 &\hfill 13 &\hfill 127 &\hfill 15\\ 
  129 &\hfill 256 &\hfill 17 &\hfill 253 &\hfill 17\\
  257 &\hfill 510 &\hfill 19 &\hfill 511 &\hfill 21\\ \hline 
  513 &\hfill 1024 &\hfill 23 &\hfill 1021 &\hfill 23\\
 1025 &\hfill 2046 &\hfill 25 &\hfill 2047 &\hfill 27\\
 2049 &\hfill 4096 &\hfill 29 &\hfill 4093 &\hfill 29\\
 4097 &\hfill 8190 &\hfill 31 &\hfill 8191 &\hfill 33\\ 
 8193 &\hfill 16384 &\hfill 35 &\hfill 16381 &\hfill 35\\ \hline
    \end{tabular}
    }
  \end{center}
\end{table}

\end{document}